\documentclass[runningheads]{llncs}

\usepackage{graphicx}
\usepackage{relsize}
\usepackage{amsmath}
\usepackage[inline]{enumitem}
\usepackage{mathtools}
\usepackage{amssymb}
\usepackage{tikz}
\usepackage[absolute]{textpos}
\usepackage{tabulary}

\usepackage{framed}
\usepackage[lined,boxed,commentsnumbered]{algorithm2e}
\usepackage{setspace}
\usepackage{url}

\newtheorem{observation}[proposition]{Observation}
\newtheorem{myclaim}{Claim}

\makeatletter
\newcommand{\removelatexerror}{\let\@latex@error\@gobble}
\makeatother

\newcommand{\floor}[1]{\left\lfloor #1 \right\rfloor}

\newcommand*{\myproofname}{Proof}

\renewenvironment{framed}[1][\hsize]
   {\MakeFramed{\hsize#1\advance\hsize-\width \FrameRestore}}%
   {\endMakeFramed}

\newcommand{\defproblem}[3]{
\vspace{1mm}
\noindent\fbox{
  \begin{minipage}{\columnwidth}
  \begin{tabular*}{\columnwidth}{@{\extracolsep{\fill}}lr} #1 \\ \end{tabular*}
  {\bf{Input:}} #2  \\
  {\bf{Question:}} #3
  \end{minipage}
}\vspace{1mm}}

\newcommand{\aAssignFull}{\textsc{\textsc{Optimal Assignment}}}
\newcommand{\bAssignFull}{\textsc{Ordered Optimal Assignment}}
\newcommand{\cAssignFull}{\textsc{Subset Optimal Assignment}}
\newcommand{\dAssignFull}{\textsc{Upper-Bounded Optimal Assignment}}
\newcommand{\vAAssignFull}{\textsc{\textsc{Verify Optimal Assignment}}}
\newcommand{\vBAssignFull}{\textsc{Verify Ordered Optimal Assignment}}
\newcommand{\vCAssignFull}{\textsc{Verify Subset Optimal Assignment}}
\newcommand{\vDAssignFull}{\textsc{Verify Upper-Bounded Optimal Assignment}}

\newcommand{\vAAssign}{\textsc{Verify-OA}}

\newcommand{\vCAssign}{\textsc{Verify-SOA}}
\newcommand{\vDAssign}{\textsc{Verify-UOA}}

\newcommand{\vAAssignShort}{\textsc{\textsc{VOA}}}

\newcommand{\vDAssignShort}{\textsc{VUOA}}

\newcommand{\aAssign}{\textsc{\textsc{OA}}}

\newcommand{\cAssign}{\textsc{SOA}}
\newcommand{\dAssign}{\textsc{UOA}}

\newcommand{\coVAAssign}{$\overline{\textsc{\textsc{Verify-OA}}}$}

\newcommand{\coVCAssign}{$\overline{\textsc{Verify-SOA}}$}
\newcommand{\coVDAssign}{$\overline{\textsc{Verify-UOA}}$}

\newcommand{\aoptimality}{optimality}
\newcommand{\boptimality}{ordered optimality}
\newcommand{\coptimality}{subset optimality}
\newcommand{\doptimality}{upper-bounded optimality}

\newcommand{\aoptimal}{optimal}
\newcommand{\boptimal}{ordered optimal}
\newcommand{\coptimal}{subset optimal}
\newcommand{\doptimal}{upper-bounded optimal}

\newcommand{\hc}{\textsc{Hamiltonian Cycle}}
\newcommand{\tildeK}{\widetilde{k}}
\newcommand{\mlassign}{\textsc{Multi-Layered Assignment}}
\newcommand{\Mod}[1]{\ (\mathrm{mod}\ #1)}

\newcommand{\Thcref}[1]{Theorem~\ref{#1}}
\newcommand{\Alcref}[1]{Algorithm~\ref{#1}}
\newcommand{\Prcref}[1]{Proposition~\ref{#1}} 
\newcommand{\Clcref}[1]{Claim~\ref{#1}} 
\newcommand{\Obcref}[1]{Observation~\ref{#1}}   
\newcommand{\Lecref}[1]{Lemma~\ref{#1}}
\newcommand{\Cocref}[1]{Corollary~\ref{#1}}
\newcommand{\Ficref}[1]{Figure~\ref{#1}}

\newcommand{\nagents}{\mathrm{\# agents}}
\newcommand{\nalloc}{\mathrm{\# alloc}}
\newcommand{\nitems}{\mathrm{\# items}}
\newcommand{\Tcref}[1]{T.~\ref{#1}}

\newcommand{\assign}{\textsc{Assignment}}

\newcommand{\MCIS}{\textsc{Multicolored Independent Set}}

\newcommand{\true}{\mathrm{true}}
\newcommand{\false}{\mathrm{false}}

\newcommand{\dd}{\mathrm{d}}
\newcommand{\cc}{\mathrm{c}}

\newcommand{\OO}{\mathcal{O}}
\newcommand{\RR}{\mathcal{R}}

\newcommand{\yes}{\textsf{Yes}}
\newcommand{\no}{\textsf{No}}
\newcommand{\yesinstance}{\yes-instance}

\newcommand{\POLY}{\textsf{P}}
\newcommand{\NP}{\textsf{NP}}
\newcommand{\coNP}{\textsf{coNP}}
\newcommand{\coNPpoly}{\textsf{coNP/poly}}
\newcommand{\FPT}{\textsf{FPT}}
\newcommand{\XP}{\textsf{XP}}
\newcommand{\WO}{\textsf{W[1]}}
\newcommand{\WOH}{\textsf{W[1]}-hard}
\newcommand{\coWOH}{\textsf{coW[1]}-hard}

\newcommand{\NPH}{\textsf{NP}-hard}
\newcommand{\CONPH}{\textsf{coNP}-hard}

\newcommand{\paraH}{\textsf{para-NP}-hard}
\newcommand{\paraCoH}{\textsf{para-coNP}-hard}
\newcommand{\ETH}{\textsf{ETH}}

\input{tikzit.sty}

\tikzstyle{vertex}=[fill=white, draw=black, shape=circle, minimum size=1.2cm]
\tikzstyle{rec}=[fill=white, draw=black, shape=rectangle, minimum width=1.4cm, minimum height=0.7cm]

\tikzstyle{edge}=[->]

\begin{document}\sloppy


\title{Verification of Multi-Layered Assignment Problems}
\author{Barak Steindl \and Meirav Zehavi}
\authorrunning{Barak Steindl and Meirav Zehavi}

\institute{Ben Gurion University of the Negev, Beer-Sheva, Israel} 

\maketitle

\begin{abstract}
The class of assignment problems is a fundamental and well-studied class in the intersection of Social Choice, Computational Economics and Discrete Allocation. In a general assignment problem, a group of agents expresses preferences over a set of items, and the task is to allocate items to agents in an ``optimal'' way. A verification variant of this problem includes an allocation as part of the input, and the question becomes whether this allocation is ``optimal''. In this paper, we generalize the verification variant to the setting where each agent is equipped with {\em multiple} incomplete preference lists: Each list (called a {\em layer}) is a ranking of items in a possibly different way according to a different criterion. 

In particular, we introduce three multi-layer verification problems, each corresponds to an optimality notion that weakens the notion of global optimality (that is, pareto optimality in multiple layers) in a different way. Informally, the first notion requires that, for each group of agents whose size is exactly some input parameter $k$, the agents in the group will not be able to trade their assigned items among themselves and benefit in at least $\alpha$ layers; the second notion is similar, but it concerns all groups of size at most $k$ rather than exactly $k$; the third notion strengthens these notions by requiring that groups of $k$ agents will not be part of possibly larger groups that benefit in at least $\alpha$ layers. We study the three problems from the perspective of parameterized complexity under several natural parameterizations such as the number of layers, the number of agents, the number of items, the number of allocated items, the maximum length of a preference list, and more. We present an almost comprehensive picture of the parameterized complexity of the problems with respect to these parameters.
\end{abstract}
\section{Introduction}
The field of resource allocation problems has been widely studied in recent years. A central class of problems in this field is the class of {\em assignment problems} \cite{RePEc:ecm:emetrp:v:66:y:1998:i:3:p:689-702,Bogomolnaia2001ANS,Abdulkadirog1999HouseAW,10.1007/978-3-540-30551-4_3,RePEc:eee:jetheo:v:52:y:1990:i:1:p:123-135,ijcai2017-12,soton425734,AAAI1817396,gourves:hal-01741519}. In the most general, abstract formulation of an assignment problem (to which we will simply refer as the {\em general assignment problem}), an instance consists of a set of $n$ agents and a set of $m$ items. Each agent (human, company, or any other entity) expresses preferences over a subset of items, and the objective is to allocate items to agents in an ``optimal'' way. A verification variant of the general assignment problem includes, in addition, some allocation as part of the input, and the question becomes whether this allocation is ``optimal''. 

Different notions of optimality have been considered in the literature, but the one that has received the most attention is {\em pareto optimality} (see, e.g., \cite{RePEc:ecm:emetrp:v:66:y:1998:i:3:p:689-702,ijcai2017-12,soton425734}). Intuitively, an assignment $p$ is called {\em pareto optimal} if there is no other assignment $q$ that is at least good as $p$ for all agents and also strictly better than $p$ for at least one agent. An equivalent requirement for an assignment to be pareto optimal is to admit no {\em trading cycle} (see, e.g., Aziz et al.~\cite{ijcai2017-12,soton425734}). Intuitively, an assignment admits a trading cycle if there exists a set of agents who all benefit by exchanging their allocated items among themselves (as indicated by the cycle). It is known to imply that the problem of verifying whether an assignment is pareto optimal can be solved efficiently in polynomial time (see, e.g., Aziz et al.~\cite{ijcai2017-12,soton425734}). Even the seemingly more difficult problem of finding a pareto optimal assignment can also be solved in polynomial time, as shown by Abdulkadiroglu and Sönmez \cite{RePEc:ecm:emetrp:v:66:y:1998:i:3:p:689-702}.

Besides their theoretical interest, these problems (both decision and verification variants) have also practical importance. Algorithms for both variants are applied in a variety of real-world situations, such as assigning jobs to workers, campus houses to students, time stamps to users on a common machine, players to sports teams, graduating medical students to their first hospital appointments, and so on. In particular, algorithms for verifying whether an assignment is optimal are useful in cases where we already have an assignment and we want to check whether it is pareto optimal; if it is not, we may seek a ``strategy'' to improve the assignment (e.g., a trading cycle). For example, when the input is large, finding an optimal assignment may be computationally expensive, so a better choice will be to use some heuristic to find an initial assignment, verify whether it is optimal or not, and proceed accordingly.

In the general assignment problem, each agent has exactly one preference list. In this paper, we consider an extension to the verification variant of the general assignment problem in which the agents may have multiple preference lists (rather than only one).  The preference lists may represent a single subjective criterion according to which each agent ranks the items. However, it may also represent a combination of different such criteria: each agent associates a score to each item per criterion, and a single preference list is derived from some weighted sum of the scores. In many cases, it is unclear how to combine scores associated with criteria of inherently incomparable nature - that is like ``comparing apples with oranges''. Additionally, even if a single list can be forcefully extracted, most data is lost. 

Thus, the classic model seems somewhat restrictive in real world scenarios where people rely on highly {\em varied} aspects to rank other entities. For example, suppose that there are $n$ candidates who need to be assigned to $n$ positions. The recruiters may rank the candidates for each position according to different criteria, such as academic background, experience, impression by the interview, and so on \cite{KINICKI1985117,alderfer1970personal}. Moreover, when assigning campus houses to students, the student may rank the houses by multiple criteria such as their location (how close the house is to their faculty), rent, size etc \cite{wu2016students}. Transferable skills such as verbal communication, adaptability, problem solving and so on may also be factors when making decisions. This motivates the employment of multiple preference lists where each preference list (called a {\em layer}) is defined by a different criterion.

In many real-world scenarios, the preferences of the agents may sometimes depend on external circumstances that may not be completely known in advance such as growth of stocks in the market, natural phenomena, outbreak of pandemics \cite{zeren2020impact,TOPCU2020101691} and so on. In such cases, each layer in our generalized model can represent a possible ``state'' of the world, and we may seek an assignment that is optimal in as many states as possible. For instance, suppose that there is a taxi firm with $n$ taxis and $m$ costumers ($n > m$) that want to be picked at a specific time in future. The ``cost'' of each taxi depends on the time taken to reach the costumer from the starting location of the taxi. Many factors (that may not be completely known a-priori) may affect the total cost such as road constructions, weather, car condition and the availability of the drivers \cite{cools2010assessing,saleh2017study}. The firm may suggest different possible scenarios (each represents a layer). For each scenario, the costumers may be ranked differently by the taxis, and an assignment that is pareto optimal in as many layers as possible will cover most of the scenarios and will give the lowest expected total cost. 

Furthermore, it is not always possible to completely take hold of preferences of some (or all) agents due to lack of information or communication, as well as security and privacy issues \cite{nass2009value,browne2000lives}. In addition, even if it is technically and ethically feasible, it may be costly in terms of money, time, or other resources to gather all information from all the agents \cite{mulder2019willingness}. In these cases, we can ``complete the preferences'' using different assumptions on the agents. As a result, we will have a list of preference profiles that represent different possible states of the world. An assignment that is pareto optimal in as many preference profiles as possible will be pareto optimal with high probability.

Our work is inspired by the work of Chen et al.~\cite{10.1145/3219166.3219168}, who studied the \textsc{Stable Marriage} problem under multiple preference lists. In contrast to assignment problems, the \textsc{Stable Marriage} problem is a {\em two-sided} matching problem, i.e.~it consists of two disjoint sets of agents $A$ and $B$, such that each agent strictly ranks the agents of the opposite set (in assignment problems, only the agents rank the items. The objective in the \textsc{Stable Marriage} problem is to find a matching (called a {\em stable matching}) between $A$ and $B$ such that there do not exist agents $a \in A$ and $b \in B$ that are not matched to each other but rank each other higher than their matched partners. Chen et al.~\cite{10.1145/3219166.3219168} considered an extension of the \textsc{Stable Marriage} problem where there are $\ell$ layers of preferences, and adapted the definition of stability accordingly. Specifically, three notions of stability were defined: {\em $\alpha$-global stability}, {\em $\alpha$-pair stability}, and {\em $\alpha$-individual stability}. In their work, Chen et al.~\cite{10.1145/3219166.3219168} studied the algorithmic complexity of finding matchings that satisfy each of these stability notions.

In our recent work \cite{steindl2020parameterized}, we defined the notion of {\em global optimality}, which (similarly to global stability defined by Chen et al.~\cite{10.1145/3219166.3219168}) extends the notion of pareto optimality to the case where there are multiple layers by requiring an assignment to be pareto optimal in a given number of layers. They studied the parameterized complexity of the problem of finding a globally optimal assignment (in the presence of multiple preference lists), and they showed that it is an extremely hard computational task with no efficient parameterized algorithms (with respect to almost any parameter combination). Two factors cause this hardness: First, in general, finding an optimal assignment is harder than verifying whether an assignment is optimal. Second, the concept of global optimality, which requires ``global agreement'' among the agents on the layers where beneficial trading cannot be performed, may seem too strong. Thus, a natural direction is to consider an adaptation of the verification variant to the multi-layer model, and to weaken the notion of global optimality.

We define three new notions of optimality: {\em $(k,\alpha)$-\aoptimality}, {\em $(k,\alpha)$-\doptimality}, and {\em $(k,\alpha)$-\coptimality}. Intuitively, the first notion requires that each subset of agents of size $k$ does not admit trading cycles (without additional agents) in at least $\alpha$ layers. The second notion is similar, but it additionally applies this condition on all subsets of size at most $k$ rather than only exactly $k$. The third notion requires that each subset of $k$ does not appear in trading cycles, together with possibly other agents, in at least $\alpha$ layers. In contrast to the notion of global optimality, these notions do not require having the same $\alpha$ layers where all agents cannot trade and benefit - each ``small'' subset of agents may have different $\alpha$ layers where the agents in this subset do not admit trading cycles.

The consideration of the parameter $k$ is reasonable: indeed, suppose that some assignment can be improved but only if a large group of agents would exchange their items among themselves. In many cases, such trading may not be simple or feasible since it may require a lot of efforts and organization \cite{dickerson2016position,dickerson2012optimizing}. Thus, we define $k$ as a fixed size or as an upper bound on the size of agent groups for which trading can be performed (in $(k,\alpha)$-\aoptimal\ and $(k,\alpha)$-\doptimal). In contrast, in the definition of $(k,\alpha)$-\coptimality, the parameter $k$ is, essentially, a ``lower bound'' on the size of subsets which do not admit trading cycles. This notion was designed to represent real scenarios where (i) we are not interested in finding short trading cycles but only in finding large ones since they would gain the most benefit and only they might justify changing the status quo, and where (ii) large and complicated trading cycles can be performed \cite{ashlagi2012need}.

Although the verification variant of the assignment problem can be easily solved in polynomial time in the classic single-layer model (see, e.g., Aziz et al.~\cite{ijcai2017-12,soton425734}), similarly to the decision variant, the problem becomes harder when multiple preference lists are taken into account. However, we show that, while some verification variants are still hard with respect to various parameters, they also admit fixed-parameter algorithms rather than mainly hardness results. Our results show that the new verification variants are, in general, much easier to solve than the decision variant in \cite{steindl2020parameterized}. For instance, the verification variants are solvable in $\OO^{*}(2^{n})$ time, but the decision variant is proved to admit no $\OO^{*}(2^{o(k \log{(k)})})$-time algorithm even for $k = n + m + \alpha$ and for $k = n + m + (\ell - \alpha)$ (where $n = \nagents$ and $m = \nitems$), unless Exponential Time Hypothesis (\ETH) fails. In addition, the decision variant is \WOH\ when parameterized by $m + \alpha$ and $m + (\ell - \alpha)$, while the verification variants are $\FPT$ with respect to $m$, $n$, and even $\nalloc$ (the number of allocated items). Furthermore, the decision variant is unlikely to admit a polynomial kernel with respect to $m + \ell$, whereas the verification variants admit polynomial kernels with respect to $m +\ell$ and $\nalloc +\ell$. 


\begin{table*}[t]
\centering
\small\addtolength{\tabcolsep}{-2pt}
 \begin{tabulary}{\textwidth}{L | C | C | C}
Parameter & Complexity Class & Running Time & \:\:Polynomial Kernel? \\
\hline
$\alpha$ & \paraCoH\ (\Tcref{theorem:coNPH}) & - & -\\
&  \vDAssignShort: \POLY\ when $\alpha = \ell$ & polynomial & yes\\
$\ell$ & \paraCoH\ (\Tcref{theorem:coNPH}) & - & -\\
$k$ & \coWOH\ (\Tcref{theorem:coWOH}) & $\OO^{*}(n^{\OO(k)})$ (\Tcref{theorem:xpAlg}) & -\\
 &  \vAAssignShort, \vDAssignShort: \XP\ (\Tcref{theorem:xpAlg}) & $\OO^{*}(n^{\OO(k)})$ (\Tcref{theorem:xpAlg}) & -\\
$k + \ell$ & \coWOH\ (\Tcref{theorem:coWOH})  & $\OO^{*}(n^{\OO(k)})$ (\Tcref{theorem:xpAlg}) & -\\
& \vAAssignShort,\vDAssignShort: \XP\ (\Tcref{theorem:xpAlg}) & $\OO^{*}(n^{\OO(k)})$ (\Tcref{theorem:xpAlg}) & -\\
$k + d$ & \:\:\:\vAAssignShort,\vDAssignShort: \FPT\ (\Tcref{theorem:fptDPlusK})\:\:\: & $\OO^{*}(d^{k})$ (\Tcref{theorem:fptDPlusK}) & no (\Tcref{theorem:cc1})\\
$(n-k) + \ell + d$ & \paraCoH\ (\Tcref{theorem:coNPH}) & - & -\\
$\nalloc$ & \FPT\ (\Tcref{theorem:fptAlg}) & \:\:$\OO^{*}(2^{\nalloc})$ (\Tcref{theorem:fptAlg})\:\:& no (\Tcref{theorem:cc1})\\
$\nalloc + \ell$ & \FPT\ (\Tcref{theorem:fptAlg}) & $\OO^{*}(2^{\nalloc})$ (\Tcref{theorem:fptAlg})& yes (\Tcref{theorem:kernelAlloc})\\
$n+ m+ \alpha$ & \FPT\ (\Tcref{theorem:fptAlg}) & $\OO^{*}(2^{\nalloc})$ (\Tcref{theorem:fptAlg})& no (\Tcref{theorem:cc1})\\
$n + m+ (\ell - \alpha)$\:\:\:& \FPT\ (\Tcref{theorem:fptAlg}) & $\OO^{*}(2^{\nalloc})$ (\Tcref{theorem:fptAlg})& no (\Tcref{theorem:cc1})
\end{tabulary}

\caption{Summary of our results for the problems \vAAssignFull, \vDAssignFull, and \vCAssignFull. The results are applicable to the three problems, unless stated otherwise.}
\label{tab:results-summary}
\end{table*}

\smallskip
\noindent\textbf{Our Contributions.}
We consider several parameters such as the number of layers $\ell$, the number of agents $n = \nagents$, the number of items $m = \nitems$, the maximum length of a preference list $d$, the number of allocated items $\nalloc$, and the parameters $\alpha$ and $k$ that are related to the optimality concepts (see Section \ref{sec:preliminaries} for the formal definitions). In particular, we present an almost comprehensive picture of the parameterized complexity of the problems with respect to these parameters. 

The choice of these parameters is sensible because in real-life scenarios such as those mentioned earlier, some of these parameters may be substantially smaller than the input size. For instance, $\ell$ and $\alpha$ are upper bounded by the number of criteria according to which the agents rank the items. Thus, they are likely to be small in practice: when ranking other entities, people usually do not consider a substantially large number of criteria. For instance, when sports teams rank candidate players, only a few criteria such as the player's winning history, his impact on his previous teams, and physical properties are taken into account \cite{fearnhead2011estimating}. Moreover, in various cases concerning ranking of people, jobs, houses etc., people usually have a limited number of entities that they want or are allowed to ask for \cite{clinedinst2019state}. In these cases, the parameter $d$ is likely to be small. In addition, in small countries (such as Israel), the number of universities, hospitals, sports teams and many other facilities and organizations is very small \cite{chernichovskystate,educationInIsrael}. Thus, in scenarios concerning these entities, at least one among $n$ and $m$ (and thus also $\nalloc$) may be small. Furthermore, when assigning students to universities, workers to companies (or to work teams in a company), and players to sports teams since the number of universities, companies (or work teams) and sports teams are usually substantially smaller than the number of students, workers and players, respectively. The consideration of $k$ is justified by previous arguments.  A summary of our results is given in Table \ref{tab:results-summary}.

\smallskip
\noindent\textbf{\coNP-Hardness.}
We first provide some simple properties of the problems \vAAssignFull\ (\vAAssign), \vDAssignFull\ (\vDAssign), and \vCAssignFull\ (\vCAssign). Afterward, we prove that \vDAssign\ is solvable in polynomial time when $\alpha = \ell$. We also assert that the three problems are in \coNP\ by providing a witness for each \yes-instance. After that, we prove that the problems are \paraCoH\ with respect to the parameter $\ell + d + (n-k)$. This is done using a polynomial reduction from the \hc\ problem on directed graphs with maximum degree $3$ (proved to be \NPH\ by Plesńik \cite{PLESNIK1979199}) to the complements of \vAAssign\ and \vCAssign. In the reduction, we construct an instance with $\ell = 1$ layers, $d = 3$, $\alpha = 1$, and $k = n$ (we show that the problems are equivalent in this case). We then extend this proof by adding another layer in order to capture \vDAssign\ as well. Informally, given a directed graph $G$ with $n$ vertices, the reduction constructs an instance of the problem with $n$ agents, $n$ items, consisting of two layers, such that the trading graph of one layer is derived by the graph $G$, and the trading graph of the second layer is a single cycle containing all the agents and items.

\smallskip
\noindent\textbf{Kernelization.}
We prove that the problems admit polynomial kernels when parameterized by $\nalloc + \ell$. Informally, given an instance of the problems, the kernels first perform a preprocessing step to verify that no agent admits self loops in many layers; then, they remove from the instance all the agents
and items which are not matched by the assignment since they cannot appear in trading cycles. Thus, we conclude that the problems admit polynomial kernels with respect to the parameters $n + \ell$ and $m + \ell$ as well. We prove that \vAAssign\ and \vCAssign\ do not admit polynomial kernels with respect to the $n + m + \alpha$ and $n + m + (\ell - \alpha)$ by providing two cross-compositions from \hc\ on directed graphs with maximum degree $3$, which rely on the reduction from Section \ref{sec:conphardness}. We then extend these cross-compositions to have the same results for \vDAssign. 

\smallskip
\noindent\textbf{Fixed-Parameter Tractability.}
We first prove that the three problems are \FPT\ with respect to $\nalloc$ (and thus, also with respect to $n$, and $m$) by providing $\OO^{*}(2^{\nalloc})$-time dynamic programming algorithms that are inspired by the technique by Björklund et al.~\cite{10.1145/1250790.1250801} to compute the Fast zeta and Möbius transform and by the Floyd–Warshall algorithm \cite{floyd62shortestpath}. We then prove that \vAAssign\ and \vDAssign\ are \XP\ with respect to the parameter $k$ by providing an $\OO^{*}(n^{\OO(k)})$-time algorithm. Informally speaking, the algorithm verifies whether each subset of agents of size $k$ (or at most $k$ for \vDAssign) does not admit ``conflicts'' (i.e.~trading cycles) in many layers, and it relies on an $\OO^{*}(2^{n})$-time algorithm for \hc\ on directed graphs by Bellman \cite{10.1145/321105.321111}. Then, we prove that \vAAssign\ and \vDAssign\ are \FPT\ with respect to the parameter $k + d$ by providing an $\OO^{*}(d^{k})$-time algorithm. The algorithm relies on the fact that there are at most $\OO(n \cdot d^{k})$ trading cycles with at most $k$ agents in each layer. It first runs the kernelization algorithm of size $\ell \cdot (\nalloc)^{2}$, and then it considers all possible trading cycles with at most $k$ agents in all the layers. For each such trading cycle, it checks whether the agents in the cycle admit trading cycle also in $\ell - \alpha$ other layers - if this is the case, it returns \no. Finally, we prove that the three problems are \coWOH\ when parameterized by $k + \ell$ using a parameterized reduction from the \MCIS\ problem to the complements of the problems. Roughly speaking, given a graph $G$ with $n$ vertices and a coloring $c$ that colors the vertices in $G$ with $\tildeK$ colors, the reduction creates an agent and an item for each vertex; In the resulting instance, $k = \tildeK$, $\alpha = 1$, $\ell = {\tildeK \choose 2}$, and we have that the agents and items that correspond to the vertices in the multicolored independent set admit trading cycles in all the layers.
\section{Preliminaries}
\label{sec:preliminaries}
For any $t \in \mathbb{N}$, let $[t] = \lbrace 1,\ldots,t \rbrace$. We use the $\OO^{*}$-notation to suppress polynomial factors in the input size, that is, $\OO^{*}(f(k)) = f(k) \cdot n^{\OO(1)}$.

\subsection{Assignment Problems} An instance of the (general) assignment problem is a triple $(A,I,P)$ where $A$ is a set of $n$ agents $\lbrace a_{1},\ldots,a_{n} \rbrace$, $I$ is a set of $m$ items $\lbrace b_{1},\ldots,b_{m} \rbrace$, and $P=(<_{a_{1}},\ldots,<_{a_{n}})$, called the {\em preference profile}, contains the preferences of the agents over the items, where each $<_{a_{i}}$ encodes the preferences of $a_{i}$ and is a linear order over a {\em subset} of $I$ (preferences are allowed to be incomplete). We refer to such linear orders as {\em preference lists}. If $b_{j} <_{a_{i}} b_{r}$, we say that agent $a_{i}$ {\em prefers} item $b_{r}$ over item $b_{j}$, and we write $b_{j} \leq_{a_{i}} b_{r}$ if $b_{j} <_{a_{i}} b_{r}$ or $b_{j} = b_{r}$. Item $b$ is {\em acceptable} by agent $a$ if $b$ appears in $a$'s preference list. An {\em assignment} is an allocation of items to agents such that each agent is allocated at most one item, and each item is allocated to at most one agent. Since the preferences of the agents may be incomplete, or the number of items may be smaller than the number of agents, some agents may not have available items to be assigned to. To deal with this case, a special item $b_{\emptyset}$ is defined, seen as the least preferred item of each agent, and will be used as a sign that an agent is not allocated an item. Throughout this paper, we assume that $b_{\emptyset}$ is not part of the item sets, and that it appears at the end of every preference list (we will not write $b_{\emptyset}$ explicitly in the preference lists). We formally define assignments as follows:

\begin{definition}\label{def:assignment}
Let $A=\lbrace a_{1},\ldots,a_{n} \rbrace$ be a set of $n$ agents and let $I = \lbrace b_{1},\ldots,b_{m} \rbrace$ be a set of $m$ items. A mapping $p:A \rightarrow I \cup \lbrace b_{\emptyset} \rbrace$ is called an {\em assignment} if for each $i \in [n]$, it satisfies one of the following conditions:
\begin{enumerate}
\item $p(a_{i})=b_{\emptyset}$.
\item Both $p(a_{i}) \in I$ and for each $j \in [n]\setminus \lbrace i \rbrace$, $p(a_{i}) \neq p(a_{j})$.
\end{enumerate}
\end{definition}

We refer to $p$ as {\em legal} if it satisfies $p(a_{i})=b_{\emptyset}$ or that $p(a_{i}) \in I$ is acceptable by $a_{i}$ for each $i \in [n]$. For brevity, we will omit the term ``legal'' and refer to a legal assignment just as an assignment.\footnote{All the ``optimal'' assignments that we construct in this paper will be legal for each agent group in a sufficient number of layers.}Moreover, when we write a set in a preference list, we assume that its elements are ordered arbitrarily, unless stated otherwise. In the general assignment problem, we are given such a triple $(A,I,P)$, and we seek an assignment which is ``optimal'' according to some criterion.

\smallskip
\noindent\textbf{Pareto Optimality.}
There are different ways to define optimality of assignments, but the one that received the most attention in the literature is {\em pareto optimality}. Informally speaking, an assignment $p$ is {\em pareto optimal} if there does not exist another assignment $q$ that is ``at least as good'' as $p$ for all the agents, and is ``better'' for at least one agent. It is formally defined as follows.

\begin{definition}\label{def:poAssignment}
Let $A = \lbrace a_{1},\ldots,a_{n} \rbrace$ be a set of agents, and let $I$ be a set of items. An assignment $p:A \rightarrow I \cup \lbrace b_{\emptyset} \rbrace$ is {\em pareto optimal} if there does not exist another assignment $q : A \rightarrow I \cup \lbrace b_{\emptyset} \rbrace$ that satisfies:
\begin{enumerate}
\item $p(a_{i}) \leq_{a_{i}} q(a_{i})$ for every $i \in [n]$.
\item There exists $i \in [n]$ such that $p(a_{i}) <_{a_{i}} q(a_{i})$.
\end{enumerate}
If such an assignment $q$ exists, we say that $q$ {\em pareto dominates} $p$.
\end{definition}

The \assign\ problem is a special case of the general assignment problem where the criterion of optimality is pareto optimality. 

In what follows, we first give some well-known characterizations of assignments, and then we introduce new concepts of optimality and two new multi-layered assignment problems. 

Intuitively, an assignment admits a {\em trading cycle} if there exists a set of agents who all benefit by exchanging their allocated items among themselves. For example, a simple trading cycle among two agents $a$ and $b$ occurs when agent $a$ prefers agent $b$'s item over its own item, and agent $b$ prefers agent $a$'s item over its own item. Both $a$ and $b$ would benefit from exchanging their items. Formally, a trading cycle is defined as follows.

\begin{definition}
An assignment $p$ admits a {\em trading cycle} $(a_{i_{0}},b_{j_{0}},a_{i_{1}},b_{j_{1}},\ldots,a_{i_{k-1}},b_{j_{k-1}})$ if for each $r \in \lbrace 0,\ldots,k-1 \rbrace$, we have that $p(a_{i_r})=b_{j_r}$ and $b_{j_r} <_{a_{i_{r}}} b_{j_{r+1 \Mod{k}}}$.
\end{definition}

\begin{definition}
An assignment $p$ admits a {\em self loop} if there exist an agent $a_{i}$ and an item $b_j$ such that $b_{j}$ is not allocated to any agent by $p$, and $p(a_{i}) <_{a_{i}} b_{j}$.
\end{definition}

\begin{proposition}[Folklore; see, e.g., Aziz et al.~\cite{ijcai2017-12,soton425734}]
\label{prop:po-iff-no-tc-and-sl}
An assignment $p$ is pareto optimal if and only if it does not admit trading cycles and self loops.
\end{proposition}

By this proposition, the problem of checking whether an assignment admits trading cycles or self loops can be reduced to the problem of checking whether the directed graph defined next contains cycles. For an instance $(A,I,P)$ and an assignment $p$, the corresponding {\em trading graph} is the directed graph defined as follows. Its vertex set is $A \cup I$, and there are three types of edges:
\begin{itemize}
\item For each $a \in A$ such that $p(a) \neq b_{\emptyset}$, there is a directed edge from $p(a)$ to $a$. Namely, each allocated item points to its owner.
\item For each agent $a \in A$, there is an edge from $a$ to all the items it prefers over its assigned item $p(a)$ (if $p(a)=b_{\emptyset}$, $a$ points to all its acceptable items).
\item Each item with no owner points to all the agents that accept it.
\end{itemize} 

\begin{proposition}[Folklore; see, e.g., Aziz et al.~\cite{ijcai2017-12,soton425734}]\label{prop:po-iff-no-cycles-in-td}
An assignment $p$ is pareto optimal if and only if its corresponding trading graph does not contain cycles.
\end{proposition}

\smallskip
\noindent\textbf{Example.} Suppose that $A = \lbrace a_{1},a_{2},a_{3},a_{4},a_{5} \rbrace$ and $I = \lbrace b_{1},b_{2},b_{3},b_{4},b_{5} \rbrace$. Assume that the preferences of the agents are defined as follows.

\begin{framed}[0.6\columnwidth]
\begin{itemize}
  \item $a_{1}\ $: $\ b_{4}\ >\ b_{1}\ >\ b_{2}\ >\ b_{5}$
  \item $a_{2}\ $: $\ b_{1}\ >\ b_{4}\ >\ b_{5}$
  \item $a_{3}\ $: $\ b_{2}\ >\ b_{1}$
  \item $a_{4}\ $: $\ b_{3}\ >\ b_{5}$
  \item $a_{5}\ $: $\ b_{5}$
\end{itemize}
\end{framed}

Let $p: A \rightarrow I \cup \lbrace b_{\emptyset} \rbrace$ be an assignment such that $p(a_{1}) = b_{2}$, $p(a_{2}) = b_{4}$, $p(a_{3}) = b_{1}$, $p(a_{4}) = b_{5}$, and $p(a_{5}) = b_{\emptyset}$.
The trading graph of the preference profile with respect to $p$ is:

\smallskip
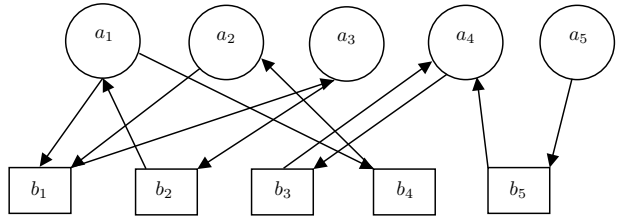
\begin{center}
\scalebox{0.75}{
\tikzset{every picture/.style={line width=0.75pt}} 

\begin{tikzpicture}[x=0.75pt,y=0.75pt,yscale=-1,xscale=1]

\draw   (182,45) .. controls (182,31.19) and (193.19,20) .. (207,20) .. controls (220.81,20) and (232,31.19) .. (232,45) .. controls (232,58.81) and (220.81,70) .. (207,70) .. controls (193.19,70) and (182,58.81) .. (182,45) -- cycle ;
\draw   (99,44) .. controls (99,30.19) and (110.19,19) .. (124,19) .. controls (137.81,19) and (149,30.19) .. (149,44) .. controls (149,57.81) and (137.81,69) .. (124,69) .. controls (110.19,69) and (99,57.81) .. (99,44) -- cycle ;
\draw   (61,129.2) -- (102.5,129.2) -- (102.5,160) -- (61,160) -- cycle ;
\draw    (153,130) -- (125.29,71.71) ;
\draw [shift={(124,69)}, rotate = 424.57] [fill={rgb, 255:red, 0; green, 0; blue, 0 }  ][line width=0.08]  [draw opacity=0] (8.93,-4.29) -- (0,0) -- (8.93,4.29) -- cycle    ;
\draw    (102.5,129.2) -- (276.65,71.28) ;
\draw [shift={(279.5,70.33)}, rotate = 521.6] [fill={rgb, 255:red, 0; green, 0; blue, 0 }  ][line width=0.08]  [draw opacity=0] (8.93,-4.29) -- (0,0) -- (8.93,4.29) -- cycle    ;
\draw    (306,130.2) -- (232.63,57.77) ;
\draw [shift={(230.5,55.67)}, rotate = 404.63] [fill={rgb, 255:red, 0; green, 0; blue, 0 }  ][line width=0.08]  [draw opacity=0] (8.93,-4.29) -- (0,0) -- (8.93,4.29) -- cycle    ;
\draw   (343,45) .. controls (343,31.19) and (354.19,20) .. (368,20) .. controls (381.81,20) and (393,31.19) .. (393,45) .. controls (393,58.81) and (381.81,70) .. (368,70) .. controls (354.19,70) and (343,58.81) .. (343,45) -- cycle ;
\draw   (263,46) .. controls (263,32.19) and (274.19,21) .. (288,21) .. controls (301.81,21) and (313,32.19) .. (313,46) .. controls (313,59.81) and (301.81,71) .. (288,71) .. controls (274.19,71) and (263,59.81) .. (263,46) -- cycle ;
\draw   (146,130.2) -- (187.5,130.2) -- (187.5,161) -- (146,161) -- cycle ;
\draw   (224,130.2) -- (265.5,130.2) -- (265.5,161) -- (224,161) -- cycle ;
\draw   (306,130.2) -- (347.5,130.2) -- (347.5,161) -- (306,161) -- cycle ;
\draw   (383,129.2) -- (424.5,129.2) -- (424.5,160) -- (383,160) -- cycle ;
\draw    (124,69) -- (83.24,126.22) ;
\draw [shift={(81.5,128.67)}, rotate = 305.46] [fill={rgb, 255:red, 0; green, 0; blue, 0 }  ][line width=0.08]  [draw opacity=0] (8.93,-4.29) -- (0,0) -- (8.93,4.29) -- cycle    ;
\draw    (148.5,52.33) -- (303.31,128.87) ;
\draw [shift={(306,130.2)}, rotate = 206.31] [fill={rgb, 255:red, 0; green, 0; blue, 0 }  ][line width=0.08]  [draw opacity=0] (8.93,-4.29) -- (0,0) -- (8.93,4.29) -- cycle    ;
\draw    (189.5,62.33) -- (104.88,127.37) ;
\draw [shift={(102.5,129.2)}, rotate = 322.45] [fill={rgb, 255:red, 0; green, 0; blue, 0 }  ][line width=0.08]  [draw opacity=0] (8.93,-4.29) -- (0,0) -- (8.93,4.29) -- cycle    ;
\draw    (279.5,70.33) -- (190.01,128.56) ;
\draw [shift={(187.5,130.2)}, rotate = 326.95] [fill={rgb, 255:red, 0; green, 0; blue, 0 }  ][line width=0.08]  [draw opacity=0] (8.93,-4.29) -- (0,0) -- (8.93,4.29) -- cycle    ;
\draw    (355.5,66.33) -- (267.95,128.46) ;
\draw [shift={(265.5,130.2)}, rotate = 324.64] [fill={rgb, 255:red, 0; green, 0; blue, 0 }  ][line width=0.08]  [draw opacity=0] (8.93,-4.29) -- (0,0) -- (8.93,4.29) -- cycle    ;
\draw    (245.5,129.67) -- (343.07,59.42) ;
\draw [shift={(345.5,57.67)}, rotate = 504.25] [fill={rgb, 255:red, 0; green, 0; blue, 0 }  ][line width=0.08]  [draw opacity=0] (8.93,-4.29) -- (0,0) -- (8.93,4.29) -- cycle    ;
\draw    (383,129.2) -- (375.87,71.31) ;
\draw [shift={(375.5,68.33)}, rotate = 442.98] [fill={rgb, 255:red, 0; green, 0; blue, 0 }  ][line width=0.08]  [draw opacity=0] (8.93,-4.29) -- (0,0) -- (8.93,4.29) -- cycle    ;
\draw   (418,45) .. controls (418,31.19) and (429.19,20) .. (443,20) .. controls (456.81,20) and (468,31.19) .. (468,45) .. controls (468,58.81) and (456.81,70) .. (443,70) .. controls (429.19,70) and (418,58.81) .. (418,45) -- cycle ;
\draw    (439.5,69.33) -- (425.23,126.29) ;
\draw [shift={(424.5,129.2)}, rotate = 284.07] [fill={rgb, 255:red, 0; green, 0; blue, 0 }  ][line width=0.08]  [draw opacity=0] (8.93,-4.29) -- (0,0) -- (8.93,4.29) -- cycle    ;

\draw (117,35) node [anchor=north west][inner sep=0.75pt]   [align=left] {$\displaystyle a_{1}$};
\draw (199,36) node [anchor=north west][inner sep=0.75pt]   [align=left] {$\displaystyle a_{2}$};
\draw (360,36) node [anchor=north west][inner sep=0.75pt]   [align=left] {$\displaystyle a_{4}$};
\draw (280,37) node [anchor=north west][inner sep=0.75pt]   [align=left] {$\displaystyle a_{3}$};
\draw (74,137) node [anchor=north west][inner sep=0.75pt]   [align=left] {$\displaystyle b_{1}$};
\draw (158,137) node [anchor=north west][inner sep=0.75pt]   [align=left] {$\displaystyle b_{2}$};
\draw (238,137) node [anchor=north west][inner sep=0.75pt]   [align=left] {$\displaystyle b_{3}$};
\draw (320,137) node [anchor=north west][inner sep=0.75pt]   [align=left] {$\displaystyle b_{4}$};
\draw (396,137) node [anchor=north west][inner sep=0.75pt]   [align=left] {$\displaystyle b_{5}$};
\draw (436,36) node [anchor=north west][inner sep=0.75pt]   [align=left] {$\displaystyle a_{5}$};

\end{tikzpicture}
}
\end{center}

Observe that agents $a_{1}$, $a_{2}$ and $a_{3}$ admit the trading cycle $(a_{1},b_{2},a_{2},b_{4},a_{3},b_{1})$, and agent $a_{4}$ admits a self loop with $b_{3}$. By \Prcref{prop:po-iff-no-cycles-in-td}, $p$ is not pareto optimal. If $a_{1}$, $a_{2}$, and $a_{3}$ exchange their items, $a_{4}$ gets $b_{3}$, and $a_{5}$ gets $b_{5}$, we have a pareto optimal assignment $q$ in which $q(a_{1}) = b_{4}$, $q(a_{2}) = b_{1}$, $q(a_{3}) = b_{2}$, $q(a_{4}) = b_{3}$ and $q(a_{5}) = b_{5}$.

\subsection{Generalization of the Assignment Problem} We introduce a generalized assignment problem where there are $\ell$ layers of preferences. For each $j \in [\ell]$, we refer to $<_{a_{i}}^{(j)}$ as $a_{i}$'s preference list in layer $j$. The {\em preference profile in layer $j$} is the collection of the agents' preference lists in the layer, namely, $P_{j}=(<_{a_{1}}^{(j)},\ldots,<_{a_{n}}^{(j)})$. Thus, the new problem is defined as follows.

\noindent\scalebox{0.92}{
\defproblem {$\mlassign$}{$(A,I,P_{1},\ldots,P_{\ell})$, where $A$ is a set of $n$ agents, $I$ is a set of $m$ items, $P_{i}$ is the preference profile in layer $i$ for each $i \in [\ell]$.}{Does an ``optimal'' assignment exist?}
}

\smallskip
\noindent\textbf{New Concepts of Optimality.} 
In \cite{steindl2020parameterized}, we introduced a new concept of optimality that naturally extends pareto optimality by requiring an assignment to be pareto optimal in a given number of layers. This criterion is formally defined as follows.

\begin{definition} \label{def:alpha-globally-optimal}
An assignment $p$ is {\em $\alpha$-globally optimal} for an instance $(A,I,P_{1},\ldots,P_{\ell})$ if there exist $\alpha$ layers $i_{1},\ldots ,i_{\alpha} \in [\ell]$ such that $p$ is pareto optimal in the single layered instance $(A,I,P_{i_{j}})$, for each $j \in [\alpha]$.
\end{definition}

The corresponding decision problem was defined as follows:

\noindent\scalebox{0.94}{
\defproblem {\textsc{Globally Optimal Assignment}}{$(A,I,P_{1},\ldots,P_{\ell},\alpha)$, where $A$ is a set of $n$ agents, $I$ is a set of $m$ items, $P_{i}$ is the preference profile in layer $i$ for each $i \in [\ell]$, and $\alpha \in [\ell]$.}{Does an $\alpha$-globally optimal assignment exist?}
}

In this paper, we ``weaken'' this notion. By \Prcref{prop:po-iff-no-tc-and-sl}, with respect to an $\alpha$-globally optimal assignment, there exist $\alpha$ layers with no trading cycles or self loops (and therefore all the agents cannot exchange items and benefit in these layers). In other words, there is a ``global agreement'' among all the agents on the layers where they cannot benefit by trading. This requirement may seem too strong, thus we weaken it as follows. Instead of requiring the same $\alpha$ layers for all the agents, we will require that each group of agents of a bounded size will have its own $\alpha$ layers where the agents in the group cannot exchange their items and benefit; the layers for each group may be different. 

\begin{definition}\label{def:tradingCycleGroupInLayer}
Let $(A,I,P_{1},\ldots,P_{\ell})$ be an instance of the $\mlassign$ problem, let $p : A \rightarrow I \cup \lbrace b_{\emptyset} \rbrace$ be an assignment, and let $K \subseteq A$ be a subset of agents. We say that {\em $K$ admits a trading cycle with respect to $p$ in layer $j \in [\ell]$} if all the agents in $K$ appear together (with no additional agents) in a trading cycle with respect to $p$ in the single layered instance $(A,I,P_{j})$. For $a \in A$, we say that $a$ {\em admits a self loop with respect to $p$} in layer $j$ if it admits a self loop with respect to $p$ in the single layered instance $(A,I,P_{j})$.
\end{definition}

We now define three new concepts of optimality.

\begin{definition} [{\em $(k,\alpha)$-\aoptimality}] \label{def:AOptimality}
An assignment $p$ is {\em $(k,\alpha)$-\aoptimal} for an instance $(A,I,P_{1},\ldots,P_{\ell})$ if it satisfies the following conditions: 
\begin{enumerate}
\item If $k \ge 2$, then for each subset of agents $K \subseteq A$ such that $|K| = k$, there exist $\alpha$ layers $i_{1},\ldots,i_{\alpha}$ such that $K$ does not admit a trading cycle in layer $i_{j}$ with respect to $p$, for each $j \in [\alpha]$.
\item If $k = 1$, then for each $a \in A$, there exist $\alpha$ layers where $a$ does not admit a self loop.
\end{enumerate}
\end{definition}

Intuitively, the first concept requires that each group of $k$ agents cannot trade and benefit in at least $\alpha$ layers (that may depend on the specific group). So, no group of $k$ agents has incentive to ``rebel'' against the given assignment. In order to show that an assignment $p$ is not $(k,\alpha)$-\aoptimal, we will usually rely on the following equivalent definition:

\begin{observation}\label{observation:notAOptimalA}
An assignment $p$ is not $(k,\alpha)$-\aoptimal\ for an instance $(A,I,P_{1},\ldots,P_{\ell})$ if it satisfies one of the following: 
\begin{enumerate}
    \item If $k \ge 2$, then there exist a subset $K \subseteq A$ of agents such that $|K| = k$, and $\ell - \alpha + 1$ layers $i_{1},\ldots,i_{\ell - \alpha + 1}$ such that $K$ admits a trading cycle in layer $i_{j}$, for each $j \in [\ell - \alpha +1]$.
    \item If $k = 1$, then there exists $a \in A$ that appears in self loops in some $\ell - \alpha + 1$ layers.
\end{enumerate}
\end{observation}

The second definition strengthens this definition by not allowing groups of size at most $k$ to admit conflict in a large number of layers. It is formally defined as follows.

\begin{definition} [{\em $(k,\alpha)$-\doptimality}] \label{def:DOptimality}
An assignment $p$ is {\em $(k,\alpha)$-\doptimal} for an instance $(A,I,P_{1},\ldots,P_{\ell})$ if it satisfies the following conditions: 
\begin{enumerate}
\item For each subset of agents $K \subseteq A$ such that $|K| \leq k$, there exist $\alpha$ layers $i_{1},\ldots,i_{\alpha}$ such that $K$ does not admit a trading cycle in layer $i_{j}$ with respect to $p$, for each $j \in [\alpha]$.
\item For each $a \in A$, there exist $\alpha$ layers where $a$ does not admit a self loop.
\end{enumerate}
\end{definition}

Clearly, an assignment is $(k,\alpha)$-\doptimal\ if and only if it is $(k',\alpha)$-\aoptimal\ for all $k' \in [k]$ (simultaneously). So, algorithms regarding the first notion of optimality can be used in the context of the second. However, the problems may not be computationally equivalent, as the second may be easier (in various aspects) than the first, and hence both are defined explicitly. Indeed, we have similar ``equivalence'' between \textsc{$k$-Cycle} (finding a cycle of length exactly $k$ in a given graph) and \textsc{$k$-Shortest Cycle} (finding the shortest cycle of length at most $k$ in a given graph), while the first problem among these two is known to be \NPH\ (since it generalizes \hc), and the second is in \POLY\ (these problems may not be equivalent in different senses, such as solvability in polynomial time, kernelization complexity, and more). In order to show that an assignment is not $(k,\alpha)$-\doptimal, similarly to \Obcref{observation:notAOptimalA}, we can show that there exists an agent which admits self loops in some $\ell - \alpha + 1$ layers, or a subset of agents which admits trading cycles in some $\ell - \alpha + 1$ layers.

The third definition is similar to the first one, but additionally, it further does not allow small groups of agents to be part of larger trading cycles. It is formally defined as follows.

\begin{definition} [{\em $(k,\alpha)$-\coptimality}] \label{def:COptimality}
An assignment $p$ is {\em $(k,\alpha)$-\coptimal} for an instance $(A,I,P_{1},\ldots,P_{\ell})$ if it satisfies the following conditions:
\begin{enumerate}
\item For each subset of agents $K \subseteq A$ such that $|K| = k$, there exist $\alpha$ layers $i_{1},\ldots,i_{\alpha}$ such that, for each $j \in [\alpha]$, there does not exist $K' \subseteq A$ that contains $K$ ($K \subseteq K' \subseteq A$) and admits a trading cycle in layer $i_{j}$.
\item If $k=1$, then for each $a \in A$, there exist $\alpha$ layers where it does not admit a self loop.
\end{enumerate}
\end{definition}

Notice that when $k = 1$, both conditions need to be satisfied (we mentioned them separately since the notions of trading cycle and self loop are different). In order to show that an assignment $p$ is not $(k,\alpha)$-\coptimal, we will usually use the following equivalent definition:

\begin{observation}\label{observation:notAOptimalC}
An assignment $p$ is not $(k,\alpha)$-\coptimal\ for an instance $(A,I,P_{1},\ldots,P_{\ell})$ if it satisfies one of the following. 
\begin{enumerate}
    \item There exist a subset $K \subseteq A$ such that $|K| = k$, $\ell -\alpha +1$ subsets $K_{1},\ldots,K_{\ell-\alpha+1}$ such that $K \subseteq K_{i} \subseteq A$ for each $i \in [\ell - \alpha +1]$, and $\ell - \alpha + 1$ layers $i_{1},\ldots,i_{\ell - \alpha + 1}$ such that $K_{j}$ admits a trading cycle in layer $i_{j}$, for each $j \in [\ell - \alpha +1]$.
    \item If $k = 1$: There exists $a \in A$ that appears in self loops in some $\ell - \alpha + 1$ layers.
\end{enumerate}
\end{observation}

Notice that if we change the first condition in Definition \ref{def:COptimality} to apply on all subsets of size at most $k$ (rather than on subsets of size exactly $k$), and the second condition to be satisfied for any value of $k$, then the parameter $k$ becomes senseless since $(k,\alpha)$-\coptimality\ will be just equivalent to $(1,\alpha)$-\coptimality\ for any $k$ (as shown in \Lecref{lemma:ubSubsetOptimal}). We remark that this equivalence is not true for $(k,\alpha)$-\doptimality\ since there are instances in which each agent does not admit self loops in some $\alpha$ layers, while there exists a subset of agents that admits trading cycles in some $\ell - \alpha +1$ layers. 

The new decision problems, where we are given an instance as in the \mlassign\ problem and seek an ``optimal'' assignment, are defined as follows:

\noindent\scalebox{0.94}{
\defproblem {\aAssignFull\ (\aAssign)}{$(A,I,P_{1},\ldots,P_{\ell},\alpha,k)$, where $A$ is a set of $n$ agents, $I$ is a set of $m$ items, $P_{i}$ is the preference profile in layer $i$ for each $i \in [\ell]$, $\alpha \in [\ell]$, and $k \in [n]$.}{Does a $(k,\alpha)$-\aoptimal\ assignment exist?}
}

The problems \dAssignFull\ (\dAssign) and \cAssignFull\ (\cAssign) are defined analogously with respect to Definitions \ref{def:DOptimality} and \ref{def:COptimality}, respectively.

\smallskip
\noindent\textbf{Verification Variants.}
In this paper, we focus on the verification variants of the problems \aAssignFull, \dAssignFull, and \cAssignFull, in which we are additionally given an assignment, and we ask whether it is optimal.

\noindent\scalebox{0.94}{
\defproblem {\vAAssignFull\ (\vAAssign)}{$(A,I,P_{1},\ldots,P_{\ell},\alpha,k,p)$, where $A$ is a set of $n$ agents, $I$ is a set of $m$ items, $P_{i}$ is the preference profile in layer $i$ for each $i \in [\ell]$, $\alpha \in [\ell]$, $k \in [n]$ and $p$ is an assignment $p: A \rightarrow I \cup \lbrace b_{\emptyset} \rbrace$.}{Is $p$ $(k,\alpha)$-\aoptimal?}
}

The other verification problems, \vDAssignFull\ (\vDAssign) and \vCAssignFull\ (\vCAssign), are defined analogously.

Our main problems are \vAAssign\ and \vCAssign. In most of our proofs we focus on these two problems at first. Then, we explain how to extend the proof in order to have the same result on the problem \vDAssign\ as well. Informally speaking, the technique is to add more preference profiles to the constructed instances which enforce each subset of agents that admits trading cycles to have a specific size.

\smallskip
\noindent\textbf{Example.}
Consider the following instance, where the agent set is $A= \lbrace a_{1},a_{2},a_{3},a_{4},a_{5}\rbrace$, the item set is $I = \lbrace b_{1},b_{2},b_{3},b_{4}, b_{5} \rbrace$, and there are four layers, defined as follows.

Layer 1:
\begin{framed}[0.4\columnwidth]
\begin{itemize}
\item $a_{1}\ $: $\ b_{2} >b_{1}$
\item $a_{2}\ $: $\ b_{3}> b_{2}>b_{1}$
\item $a_{3}\ $: $\ b_{3}>b_{1}>b_{4}$
\item $a_{4}\ $: $\ b_{2}>b_{1}>b_{3}$
\item $a_{5}\ $: $\ b_{4}>b_{2}>b_{1}>b_{3}$
\end{itemize}
\end{framed} 
Layer 2:
\begin{framed}[.4\columnwidth]
\begin{itemize}
\item $a_{1}\ $: $\ b_{2} > b_{1}$
\item $a_{2}\ $: $\ b_{2}> b_{3}$
\item $a_{3}\ $: $\ b_{1}>b_{2}>b_{3}>b_{4}$
\item $a_{4}\ $: $\ b_{3}$
\item $a_{5}\ $: $\ b_{4}>b_{3}>b_{2}$
\end{itemize}
\end{framed} 
Layer 3:
\begin{framed}[.35\columnwidth]
\begin{itemize}
\item $a_{1}\ $: $\ b_{3} > b_{2} > b_{1}$
\item $a_{2}\ $: $\ b_{4} > b_{1}> b_{2}$
\item $a_{3}\ $: $\ b_{1} > b_{3}>b_{4}$
\item $a_{4}\ $: $\ b_{5}$
\item $a_{5}\ $: $\ b_{1}>b_{3}>b_{2}$
\end{itemize}
\end{framed}
Layer 4:
\begin{framed}[.4\columnwidth]
\begin{itemize}
\item $a_{1}\ $: $\ b_{3}>b_{1}>b_{2}$
\item $a_{2}\ $: $\ b_{1}> b_{2}$
\item $a_{3}\ $: $\ b_{4}>b_{3}$
\item $a_{4}\ $: $\ \emptyset$
\item $a_{5}\ $: $\ b_{2}>b_{3}>b_{4}>b_{1}$
\end{itemize}
\end{framed}

Consider the assignment $p:A \rightarrow I \cup \lbrace b_{\emptyset} \rbrace$ that satisfies $p(a_{1})= b_{1}$, $p(a_{2})= b_{2}$, $p(a_{3})= b_{4}$, $p(a_{4})= b_{\emptyset}$, and $p(a_{5})= b_{3}$. We set $k = 2$ and $\alpha = 2$. Let us write all the self loops and the subsets of size two that admit trading cycles and in each layer:
\begin{itemize}
\item Layer 1: $\lbrace a_{2},a_{5} \rbrace$, $\lbrace a_{3} , a_{5} \rbrace$.
\item Layer 2: $\lbrace a_{3},a_{5} \rbrace$.
\item Layer 3: $\lbrace a_{1},a_{2} \rbrace$, $\lbrace a_{1}, a_{5} \rbrace$, $a_{4}$ admits a self loop.
\item Layer 4: no trading cycles with $2$ agents or self loops.
\end{itemize}

By \Obcref{observation:notAOptimalA}, since there is no subset of size $2$ that appears in trading cycles in at least $\ell - \alpha + 1 = 4 - 2 + 1 = 3$ layers, or an agent that appears in self loops in at least $3$ layers, we conclude that $p$ is $(2,2)$-\aoptimal\ and also $(2,2)$-\doptimal.
In contrast, notice that the subset $\lbrace a_{1}, a_{2} \rbrace$ appears in trading cycles with additional agents in $3$ layers:
\begin{itemize}
\item In layer 1, the subset $\lbrace a_{1},a_{2},a_{5} \rbrace$ admits the trading cycle $(a_{1},b_{1},a_{2},b_{2},a_{5},b_{3})$.
\item In layer 3, the subset $\lbrace a_{1},a_{2},a_{3} \rbrace$ admits the trading cycle $(a_{1},b_{1},a_{2},b_{2},a_{3},b_{4})$.
\item In layer 4, the subset $\lbrace a_{1},a_{2},a_{5} \rbrace$ admits the trading cycle $(a_{1},b_{1},a_{5},b_{3},a_{2},b_{2})$.
\end{itemize}
Therefore, by \Obcref{observation:notAOptimalC}, $p$ is not $(2,2)$-\coptimal.

We study these three problems from the perspective of parameterized complexity.

\subsection{Parameterized Complexity} Let $\Pi$ be an \NPH\ problem. In the framework of Parameterized Complexity, each instance of $\Pi$ is associated with a {\em parameter} $k$. Here, the goal is to confine the combinatorial explosion in the running time of an algorithm for $\Pi$ to depend only on $k$. Formally, we say that $\Pi$ is {\em fixed-parameter tractable (\FPT)} if any instance $(I, k)$ of $\Pi$ is solvable in time $f(k)\cdot |I|^{\OO(1)}$, where $f$ is an arbitrary computable function of $k$. A weaker request is that for every fixed $k$, the problem $\Pi$ would be solvable in polynomial time. Formally, we say that $\Pi$ is {\em slice-wise polynomial (\XP)} if any instance $(I, k)$ of $\Pi$ is solvable in time $f(k)\cdot |I|^{g(k)}$, where $f$ and $g$ are arbitrary computable functions of $k$. Nowadays, Parameterized Complexity supplies a rich toolkit to design \FPT\ and \XP\ algorithms \cite{DBLP:series/txcs/DowneyF13,DBLP:books/sp/CyganFKLMPPS15,fomin2019kernelization}.

Parameterized Complexity also provides methods to show that a problem is unlikely to be \FPT. The main technique is the one of parameterized reductions analogous to those employed in classical complexity. Here, the concept of \WO-hardness replaces the one of \NP-hardness, and for reductions we need not only construct an equivalent instance in \FPT\ time, but also ensure that the size of the parameter in the new instance depends only on the size of the parameter in the original one.

\begin{definition}[{\bf Parameterized Reduction}]\label{definition:parameterized-reduction}
Let $\Pi$ and $\Pi'$ be two parameterized problems. A {\em parameterized reduction} from $\Pi$ to $\Pi'$ is an algorithm that, given an instance $(I,k)$ of $\Pi$, outputs an instance $(I',k')$ of $\Pi'$ such that:
\begin{itemize}
\item $(I,k)$ is a \yes-instance\ of $\Pi$ if and only if $(I',k')$ is a \yes-instance\ of $\Pi'$.
\item $k' \leq g(k)$ for some computable function $g$.
\item The running time is $f(k) \cdot |\Pi|^{\OO(1)}$ for some computable function $f$.
\end{itemize}
\end{definition}

If there exists such a reduction transforming a problem known to be \WOH\ to another problem $\Pi$, then the problem $\Pi$ is \WO-hard as well. Central \WOH-problems include, for example, deciding whether a nondeterministic single-tape Turing machine accepts within $k$ steps, {\sc Clique} parameterized be solution size, and {\sc Independent Set} parameterized by solution size. To show that a problem $\Pi$ is not \XP\ unless \textsf{P}=\NP, it is sufficient to show that there exists a fixed $k$ such $\Pi$ is \NPH. Then, the problem is said to be \paraH.

A companion notion to that of fixed-parameter tractability is the one of a polynomial kernel. Formally, a parameterized problem $\Pi$ is said to admit a {\em polynomial compression} if there exists a (not necessarily parameterized) problem $\Pi'$ and a polynomial-time algorithm that given an instance $(I,k)$ of $\Pi$, outputs an equivalent instance $I'$ of $\Pi'$ (that is, $(I,k)$ is a \yes-instance of $\Pi$ if and only if $I'$ is a \yes-instance of $\Pi'$) such that $|I'|\leq p(k)$ where $p$ is some polynomial that depends only on $k$. In case $\Pi'=\Pi$, we further say that $\Pi$ admits a {\em polynomial kernel}. For more information on Parameterized Complexity, we refer the reader to recent books such as \cite{DBLP:series/txcs/DowneyF13,DBLP:books/sp/CyganFKLMPPS15,fomin2019kernelization}.

\smallskip
\noindent\textbf{Non-Existence of a Polynomial Compression.} Our proof of the ``unlikely existence'' of polynomial kernels relies on the well-known notions of OR-cross-composition and AND-cross-composition. Before we present these techniques, let us define the following. Suppose that $\Sigma$ is some finite alphabet, $\Sigma^{\leq n}$ is the set of all words of length at most $n$ over $\Sigma$, and $\Sigma^{*}$ is the set of all possible words of any length over $\Sigma$.

\begin{definition}[{\bf Polynomial Equivalence Relation}]
An equivalence relation $\RR$ on the set $\Sigma^{*}$ is called a {\em polynomial equivalence relation} if the following conditions are satisfied:
\begin{enumerate}
\item There exists an algorithm that, given strings $x,y \in \Sigma^{*}$, resolves whether $(x,y) \in \RR$ in time polynomial in $|x| + |y|$.
\item $\RR$ restricted to the set $\Sigma^{\leq n}$ has at most $p(n)$ equivalence classes for some polynomial $p$.
\end{enumerate}
\end{definition}

\begin{definition} [{\bf OR-Cross-Composition}]\label{def:cross-comp} A (not necessarily parameterized) problem $\Pi$ {\em OR-cross-composes} into a parameterized problem $\Pi'$ if there exists a polynomial-time algorithm, called an {\em OR-cross-composition}, that given instances $I_1,I_2,\ldots,I_t$ of $\Pi$ for some $t \in \mathbb{N}$ that are equivalent with respect to some polynomial equivalence relation $\RR$, outputs an instance $(I,k)$ of $\Pi'$ such that the following conditions are satisfied.
\begin{itemize}
\item $k\leq p(s + \log{t})$ for some polynomial $p$.
\item $(I,k)$ is a \yes-instance of $\Pi'$ if and only if at least one of the instances $I_1,I_2,\ldots,I_t$ is a \yes-instance of $\Pi$.
\end{itemize}
\end{definition}

The notion of {\em AND-Cross-Composition} is defined similarly, but instead of requiring that the output instance is a \yes-instance\ if and only if at least one of the input instances is \yes-instance, we require that the output instance is a \yes-instance\ if and only if all the input instance should be yes instances. We will use the following proposition to prove that the problems do not admit polynomial kernels. 

\begin{proposition}[\cite{DBLP:journals/jcss/BodlaenderDFH09,DBLP:journals/siamdm/BodlaenderJK14}]\label{prop:noKern}
Let $\Pi$ be an \NP-hard (not necessarily parameterized) problem that OR-cross-composes (or AND-cross-composes) into a parameterized problem $\Pi'$. Then, $\Pi'$ does not admit a polynomial compression, unless \NP$\subseteq $\coNPpoly. 
\end{proposition}

To obtain (essentially) tight conditional lower bounds for the running times of algorithms, we rely on the well-known {\em Exponential-Time Hypothesis (\ETH)} \cite{DBLP:journals/jcss/ImpagliazzoP01,DBLP:journals/jcss/ImpagliazzoPZ01,DBLP:conf/iwpec/CalabroIP09}. To formalize the statement of \ETH, first recall that  given a formula $\varphi$ in conjuctive normal form (CNF) with $n$ variables and $m$ clauses, the task of {\sc CNF-SAT} is to decide whether there is a truth assignment to the variables that satisfies $\varphi$. In the {\sc $p$-CNF-SAT} problem, each clause is restricted to have at most $p$ literals. \ETH\ asserts that {\sc 3-CNF-SAT} cannot be solved in time $\OO(2^{o(n)})$.
\section{Properties of the Concepts of Optimality}
\label{sec:propertiesOfConcepts}
\noindent\textbf{Relations between different notions of optimality.}
We start with some simple properties regarding the notions of $(k,\alpha)$-\aoptimality, $(k,\alpha)$-\doptimality, and $(k,\alpha)$-\coptimality. Additionally, we prove that \vAAssign, \vDAssign, and \vCAssign\ are in \coNP, and that \vDAssign\ is solvable in polynomial time when $\alpha = \ell$.

We begin by considering what would have happened if that the second condition in Definition \ref{def:COptimality} is applied to all subsets of agents of size at least $k$ (rather than subsets of size exactly $k$). Intuitively, we show that it does not make sense to consider \vCAssign\ with $\geq$ instead of $=$ with respect to subset sizes, as then we just get the exact same problem. In other words, this yields an alternative definition of $(k,\alpha)$-\coptimality. 

\begin{lemma}
\label{lemma:COptimality2}
An assignment $p$ is $(k,\alpha)$-\coptimal\ for an instance $(A,I,P_{1},\ldots,P_{\ell})$ if and only if it satisfies the following conditions:
\begin{enumerate}
\item For each subset of agents $K \subseteq A$ such that $|K| \geq k$, there exist $\alpha$ layers $i_{1},\ldots,i_{\alpha}$ such that, for each $j \in [\alpha]$, there does not exist $K' \subseteq A$ that contains $K$ ($K \subseteq K' \subseteq A$) and admits a trading cycle in layer $i_{j}$.
\item If $k=1$, then for each $a \in A$, there exist $\alpha$ layers where it does not admit a self loop.
\end{enumerate}
\end{lemma}
\begin{proof}
$\Rightarrow$: Suppose that $p$ is $(k,\alpha)$-\coptimal. Let $K \subseteq A$ be a subset of agents such that $|K| \geq k$, and let $X \subseteq K$ be a subset of $K$ such that $|X| = k$. By Definition \ref{def:COptimality}, there exist $\alpha$ layers, $i_{1},\ldots,i_{\alpha}$ such that for each $j \in [\alpha]$, there does not exist $Y \subseteq A$ that contains $X$ ($X \subseteq Y$) and admits a trading cycle in layer $i_{j}$. Let $j \in [\alpha]$, and let $Y \subseteq A$ be a subset such that $K \subseteq Y$. Since $X \subseteq K \subseteq Y$, we have that $Y$ does not admit a trading cycle in layer $i_{j}$. Thus, $K$ does not appear in trading cycles with possibly other agents in these $\alpha$ layers. If $k = 1$, then $p$ satisfies the second condition as well. So, both conditions are satisfied.

\smallskip
\noindent $\Leftarrow$: Suppose that $p$ satisfies both conditions. Then, if $k=1$, by Definition \ref{def:COptimality} we have that $p$ satisfies the second condition. Let $K \subseteq A$ be a subset of agents such that $|K| = k$. By the first condition of the Lemma, there exist $\alpha$ layers where the agents in $K$ do not admit trading cycles with possibly other agents. Thus, $p$ is $(k,\alpha)$-\coptimal.\qed
\end{proof}

The next lemma shows that if we modify Definition \ref{def:COptimality} such that the first condition applies to all subsets of agents of size at most $k$ (rather than exactly $k$), and the second condition is satisfied for any value of $k$, then the parameter $k$ becomes senseless. We refer to this new optimality notion as {\em $(k,\alpha)$-subset* optimality}.

\begin{lemma}
\label{lemma:ubSubsetOptimal}
Let $(A,I,P_{1},\ldots,P_{\ell})$ be an instance of \mlassign\, and let $k \in [|A|]$ and $\alpha \in [\ell]$. Then, $p$ is $(k,\alpha)$-subset* optimal if and only if it is $(1,\alpha)$-\coptimal.
\end{lemma}
\begin{proof}
$\Rightarrow$: Suppose that $p$ is $(k,\alpha)$-subset* optimal. If $k = 1$, then clearly $p$ is also $(1,\alpha)$-\coptimal. Otherwise, let $a \in A$. By the definition of subset* optimality, $a$ does not admit trading cycles with possibly other agents in at least $\alpha$ layers. Then, by Definition \ref{def:COptimality}, $p$ is $(1 , \alpha)$-\coptimal.

$\Leftarrow$: Suppose that $p$ is $(1,\alpha)$-\coptimal. First, each agent does not admit self loops in at least $\alpha$ layers. Second, let $K \subseteq A$ be a subset of at most $k$ agents and let $a \in K$. Since $p$ is $(1,\alpha)$-\coptimal, we have that $a$ is not part of trading cycles with other agents in at least $\alpha$ layers. Since $a \in K$, this implies that $K$ is not part of larger trading cycles in these layers. Then, $p$ is $(k,\alpha)$-subset* optimal.\qed
\end{proof}

Next, by the definitions of $(k,\alpha)$-\aoptimality, $(k,\alpha)$-\coptimality\, and $(k,\alpha)$-\doptimality, we immediately obtain the two following observations.

\begin{observation}
\label{lemma:aAndDEq}
Let $(A,I,P_{1},\ldots,P_{\ell})$ be an instance of $\mlassign$ with $\ell$ layers. Let $p : A \rightarrow I \cup \lbrace b_{\emptyset} \rbrace$ be an assignment. Then, for every $k \in [|A|]$ and $\alpha \in [\ell]$, $p$ is $(k',\alpha)$-\aoptimal\ for all $k' \in [k]$ (simultaneously) if and only if $p$ is $(k,\alpha)$-\doptimal.
\end{observation}

\begin{observation}
\label{obs:COptIsAOpt}
Let $(A,I,P_{1},\ldots,P_{\ell})$ be an instance of $\mlassign$ with $\ell$ layers, let $p : A \rightarrow I \cup \lbrace b_{\emptyset} \rbrace$ be an assignment, and let $k \in [|A|]$ and $\alpha \in [\ell]$. Assume that $p$ is $(k , \alpha)$-\coptimal. Then, $p$ is $(k,\alpha)$-\aoptimal. 
\end{observation}

We next prove equivalence relations between the various notions of optimality under specific choices of $k$ and $\alpha$.

\begin{lemma}\label{lemma:lemma2}
Let $(A,I,P_{1},\ldots,P_{\ell})$ be an instance of $\mlassign$ with $\ell$ layers. Let $p : A \rightarrow I \cup \lbrace b_{\emptyset} \rbrace$ be an assignment. Then, the following properties are equivalent:
\begin{enumerate}
\item $p$ is $(k,\ell)$-\aoptimal\ for all $k \in [n]$ (simultaneously).
\item $p$ is $(n,\ell)$-\doptimal.
\item $p$ is $(1,\ell)$-\coptimal.
\item $p$ is $\ell$-globally optimal.
\end{enumerate}
\end{lemma}
\begin{proof}
We prove each equivalence separately:

\smallskip
\noindent
$(1) \Rightarrow (2)$: Immediate from \Obcref{lemma:aAndDEq}.

\smallskip
\noindent
$(2) \Rightarrow (3)$:
Assume that $p$ is $(n,\ell)$-\doptimal. First, each agent does not admit self loops in all the layers. Second, for each subset $K \subseteq A$, the agents in $K$ do not admit trading cycles in all the layers. Let $\lbrace a \rbrace \subseteq A$, and let $X \subseteq A$ such that $\lbrace a \rbrace \subseteq X$. Then, $X$ does not admit trading cycles in all the layers. Thus, we conclude that $\lbrace a \rbrace$ is not part of trading cycles in all the layers. Thus, $p$ is $(1,\ell)$-\coptimal. 

\smallskip
\noindent
$(3) \Rightarrow (4)$: Assume that $p$ is $(1,\ell)$-\coptimal. Then, $p$ does not admit self loops in any layer. Towards a contradiction, suppose that there exists a layer $i \in [\ell]$ containing the trading cycle $(a_{1},p(a_{1}),\ldots,a_{t},p(a_{t}))$. Then, we have a layer where $a_{1}$ admits a trading cycle (with other agents), a contradiction to the supposition that $p$ is $(1,\ell)$-\coptimal. Thus, all the layers do not contain self loops and trading cycles. So, by \Prcref{prop:po-iff-no-tc-and-sl}, $p$ is $\ell$-globally optimal.

\smallskip
\noindent
$(4) \Rightarrow (1)$: Assume that $p$ is $\ell$-globally optimal. Then, by \Prcref{prop:po-iff-no-tc-and-sl}, $p$ does not admit self loops and trading cycles of any length in all the layers. Thus, it is clearly $(k,\ell)$-\aoptimal\ for each $k \in [n]$. 
\end{proof}

We now give a simple observation that relates $(k,\alpha)$-\aoptimality\ to $(k,\alpha)$-\coptimality\ in a very particular setting, which will arise in our proofs.

\begin{observation}
\label{obs:eqAoptAndCopt}
Let $(A,I,P_{1},\ldots,P_{\ell})$ be an instance of $\mlassign$ with $\ell$ layers, and let $p : A \rightarrow I \cup \lbrace b_{\emptyset} \rbrace$ be an assignment. Assume that for each $i \in [\ell]$, $P_{i}$ does not contain trading cycles containing more than $k$ agents with respect to $p$. Then, $p$ is $(k,\alpha)$-\aoptimal\ if and only if it is $(k,\alpha)$-\coptimal.
\end{observation}
\begin{proof}
$\Rightarrow$: Assume that $p$ is $(k,\alpha)$-\aoptimal. First, if $k = 1$, then $p$ clearly satisfies that each agent does not admit self loops in some $\alpha$ layers. Second, let $K \subseteq A$ be a subset of agents of size $k$. By the definition of $(k,\alpha)$-\aoptimality, there exist $\alpha$ layers $i_{1},\ldots,i_{\alpha}$ such that for each $j \in [\alpha]$, $K$ does not admit a trading cycle in layer $i_{j}$. Since each preference profile does not contain trading cycles with more than $k$ agents, we have that for each $K' \subseteq A$ such that $K \subseteq K'$, $K'$ does not admit a trading cycle in layer $i_{j}$. Hence, $p$ is $(k,\alpha)$-\coptimal.

$\Leftarrow$: Immediate from \Obcref{obs:COptIsAOpt}.\qed
\end{proof}

By \Obcref{obs:eqAoptAndCopt}, and since every trading cycle can contain at most $n$ agents, we conclude the following.

\begin{corollary}
\label{col:nEquiv}
Assignment $p$ is $(n,\alpha)$-\aoptimal\ if and only if it is $(n,\alpha)$-\coptimal.
\end{corollary}

\smallskip
\noindent\textbf{Membership in \coNP\ and polynomial-time algorithms.} We now assert that all problems considered in this paper are in \coNP.

\begin{lemma}\label{lemma:lemma5}
The problems \vAAssign, \vDAssign, and \vCAssign\ are in \coNP.
\end{lemma}
\begin{proof}
We prove that the complements of the problems \coVAAssign, \coVDAssign, and \coVCAssign\ are in \NP.
Given an instance $D = (A,I,P_{1},\ldots,P_{\ell},\alpha,k,p)$:
\begin{itemize}
\item  If $D$ is a \yes-instance\ of \coVAAssign\ (or \coVDAssign), then by \Obcref{observation:notAOptimalA} the witness is of one of the following forms: 
\begin{enumerate}
\item  If $k \geq 2$: A subset $K \subseteq A$ such that $|K| = k$ (or $|K| \leq k$ for \coVDAssign), $\ell - \alpha +1$ layers $i_{1},\ldots, i_{\ell-\alpha +1}$, and $\ell - \alpha +1$ trading cycles $C_{i_{1}},\ldots,C_{i_{\ell - \alpha + 1}}$, such that the agents in $K$ admit the trading cycle $C_{i_{j}}$ in layer $i_{j}$, for each $j \in [\ell - \alpha +1]$. Given the witness, a verifier can verify in polynomial time whether the agents in $K$ admit the corresponding trading cycle in each layer $P_{i_{j}}$ with respect to $p$, for each $j \in [\ell - \alpha +1]$.
\item If $k=1$ (or for any value of $k$ only if the input is an instance of \coVDAssign): An agent $a \in A$, $\ell - \alpha + 1$ layers, $i_{1},\ldots,i_{\ell - \alpha + 1}$, and $\ell - \alpha + 1$ self loops, $L_{i_{1}},\ldots,L_{i_{\ell - \alpha + 1}}$, such that $a$ admits the self loop $L_{i_{1}}$ in layer $P_{i_{j}}$ with respect to $p$, for each $j \in [\ell - \alpha + 1]$. This can be easily verified in polynomial time.
\end{enumerate}

\item  If $D$ is a \yes-instance\ of \coVCAssign, then by \Obcref{observation:notAOptimalC} the witness is of one of the following forms: 
\begin{enumerate}
\item  A subset $K \subseteq A$ such that $|K| = k$, $\ell - \alpha +1$ layers, $i_{1},\ldots,i_{\ell - \alpha + 1}$, and $\ell - \alpha + 1 $ trading cycles, $C_{i_{1}},\ldots,C_{i_{\ell - \alpha +1}}$, such that $C_{i_{j}}$ is a trading cycle that contains all the agents from $K$ together with possibly additional agents. Given the witness, the verifier can verify in polynomial time whether for each $j \in [\ell - \alpha + 1]$ (i) all the agents from $K$ appear in $C_{i_{j}}$ and (ii) $C_{i_{j}}$ is a trading cycle in layer $i_{j}$.
\item If $k = 1$: An agent $a \in A$, $\ell - \alpha + 1$ layers, $i_{1},\ldots,i_{\ell - \alpha + 1}$, and $\ell - \alpha + 1$ self loops, $L_{i_{1}},\ldots,L_{i_{\ell - \alpha + 1}}$, such that $a$ admits the self loop $L_{i_{j}}$ in $P_{i_{j}}$ with respect to $p$ for each $j \in [\ell - \alpha + 1]$. This can be easily verified in polynomial time.
\end{enumerate}
\end{itemize}\qed
\end{proof}

Recall that the first conditions in the definitions of $(k,\alpha)$-\aoptimality\ and $(k,\alpha)$-\coptimality\ apply on subsets of agents of size exactly $k$, and in the definition of \doptimality, the first condition applies to subsets of size at most $k$. We prove now that this difference makes \vDAssign\ much easier to solve than \vAAssign\ when $\alpha = \ell$.

\begin{lemma}
Let $(A,I,P_{1},\ldots,P_{\ell})$ be an instance of $\mlassign$, let $p$ be an assignment, and let $k \in [n]$. Then:
\begin{enumerate}
\item Deciding whether $p$ is $(k , \ell)$-\doptimal\ can be done in a polynomial time.
\item Deciding whether $p$ is $(k' , \ell)$-\coptimal\ for all $k' \in [k]$ (simultaneously) takes a polynomial time.
\end{enumerate}
\end{lemma}
\begin{proof}
We provide a polynomial time algorithm for each problem.

\smallskip
\noindent\textbf{\vDAssign.} 
By Definition \ref{def:DOptimality}, we need to verify whether (i) for each subset $K \subseteq A$ such that $2 \leq |K| \leq k$, the agents in $K$ do not admit trading cycles in all the layers, and whether (ii) for each $a \in A$, $a$ does not admit self loops in all the layers. Equivalently, we need to answer whether the trading graph of each layer does not contain cycles with at most $k$ agents. Thus, we construct the trading graphs of all the layers with respect to $p$, and denote them by $G_{1},\ldots,G_{\ell}$. For each $i \in [\ell]$, we run a polynomial time algorithm on $G_{i}$ to find the length of the shortest cycle $t_{i}$. If for some $i \in [\ell]$, $t_{i} \leq 2k$, which means that there exists at least one cycle with $k$ or fewer agents in $G_{i}$, we return $\false$. Otherwise, all the $t_{i}$'s are greater than $k$ (i.e.~there does not exists a subset of agents of size at most $k$ that admits cycles in at least $\ell - \alpha + 1 = 1 $ layer), and we return $\true$.

\smallskip
\noindent\textbf{The second problem.} 
By \Obcref{observation:notAOptimalC}, we need to return $\false$ if and only if (i) there exists an agent that admits self loop in at least one layer, or (ii) there exist a subset $K \subseteq A$ such that $2 \leq |K| \leq k$, and a subset $K \subseteq K' \subseteq A$, such that the agents in $K'$ admit a trading cycle in at least one layer. Notice that one of these conditions occur if and only if at least one of the trading graphs of the layers contains a cycle (of any length). Thus, we construct the trading graphs of all the layers, $G_{1},\ldots,G_{\ell}$, and run a polynomial time algorithm on each of them to check whether they contain cycles. If at least one trading graph contains cycles, we return $\false$. Otherwise, we return $\true$.\qed
\end{proof}

The algorithms rely on the fact that it is easier to decide whether a directed graph contains a cycle of length at most $k$ than to decide whether it contains a cycle of length exactly $k$, or a cycle through specified elements (see, e.g., \cite{DBLP:journals/corr/abs-1301-1517,e6ecb3cf-9adf-4c45-92d4-28f854976a9e}).
\section{Fixed-Parameter Tractability}
\label{sec:part2FPT}

In this section, we first prove that \vAAssign, \vDAssign\ and \vCAssign\ admit polynomial kernel with respect to $\nalloc + \ell$ where $\nalloc = \text{number of items allocated by the assignment}$. Then, we prove that the three problems are \FPT\ with respect to the parameters $\nalloc$, $n = \nagents$, and $m = \nitems$, by providing $\OO^{*}(2^{\nalloc})$-time algorithms. Then, we prove that \vAAssign\ and \vDAssign\ admit an \XP\ algorithm when parameterized by $k$, and even an \FPT\ algorithm when parameterized by $d+ k$ (note that $d \leq m$, but $d$ and $n$ are incomparable). After that, we prove that the three problems are  \coWOH\ when parameterized by $k + \ell$. 

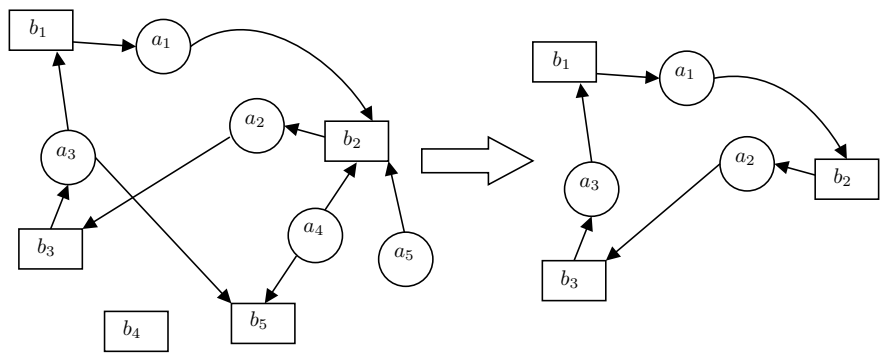
\begin{figure*}[t]
\begin{center}\scalebox{0.8}{
\tikzset{every picture/.style={line width=0.75pt}} 

\begin{tikzpicture}[x=0.75pt,y=0.75pt,yscale=-1,xscale=1]

\draw   (100,42.9) .. controls (100,33.46) and (107.66,25.8) .. (117.1,25.8) .. controls (126.54,25.8) and (134.2,33.46) .. (134.2,42.9) .. controls (134.2,52.34) and (126.54,60) .. (117.1,60) .. controls (107.66,60) and (100,52.34) .. (100,42.9) -- cycle ;
\draw   (26,158) -- (66,158) -- (66,183.2) -- (26,183.2) -- cycle ;
\draw   (159,92.9) .. controls (159,83.46) and (166.66,75.8) .. (176.1,75.8) .. controls (185.54,75.8) and (193.2,83.46) .. (193.2,92.9) .. controls (193.2,102.34) and (185.54,110) .. (176.1,110) .. controls (166.66,110) and (159,102.34) .. (159,92.9) -- cycle ;
\draw   (40,112.9) .. controls (40,103.46) and (47.66,95.8) .. (57.1,95.8) .. controls (66.54,95.8) and (74.2,103.46) .. (74.2,112.9) .. controls (74.2,122.34) and (66.54,130) .. (57.1,130) .. controls (47.66,130) and (40,122.34) .. (40,112.9) -- cycle ;
\draw   (195.8,162.1) .. controls (195.8,152.66) and (203.46,145) .. (212.9,145) .. controls (222.34,145) and (230,152.66) .. (230,162.1) .. controls (230,171.54) and (222.34,179.2) .. (212.9,179.2) .. controls (203.46,179.2) and (195.8,171.54) .. (195.8,162.1) -- cycle ;
\draw   (219,90) -- (259,90) -- (259,115.2) -- (219,115.2) -- cycle ;
\draw   (80,210) -- (120,210) -- (120,235.2) -- (80,235.2) -- cycle ;
\draw   (20,20) -- (60,20) -- (60,45.2) -- (20,45.2) -- cycle ;
\draw    (46,158) -- (55.99,132.79) ;
\draw [shift={(57.1,130)}, rotate = 471.62] [fill={rgb, 255:red, 0; green, 0; blue, 0 }  ][line width=0.08]  [draw opacity=0] (8.93,-4.29) -- (0,0) -- (8.93,4.29) -- cycle    ;
\draw    (60,40) -- (97.01,42.68) ;
\draw [shift={(100,42.9)}, rotate = 184.15] [fill={rgb, 255:red, 0; green, 0; blue, 0 }  ][line width=0.08]  [draw opacity=0] (8.93,-4.29) -- (0,0) -- (8.93,4.29) -- cycle    ;
\draw    (219,100) -- (196.09,93.7) ;
\draw [shift={(193.2,92.9)}, rotate = 375.39] [fill={rgb, 255:red, 0; green, 0; blue, 0 }  ][line width=0.08]  [draw opacity=0] (8.93,-4.29) -- (0,0) -- (8.93,4.29) -- cycle    ;
\draw   (160,205) -- (200,205) -- (200,230.2) -- (160,230.2) -- cycle ;
\draw    (201,175) -- (182.66,202.5) ;
\draw [shift={(181,205)}, rotate = 303.69] [fill={rgb, 255:red, 0; green, 0; blue, 0 }  ][line width=0.08]  [draw opacity=0] (8.93,-4.29) -- (0,0) -- (8.93,4.29) -- cycle    ;
\draw    (134.2,42.9) .. controls (173.4,13.5) and (226.33,50.66) .. (247.73,87.73) ;
\draw [shift={(249,90)}, rotate = 241.53] [fill={rgb, 255:red, 0; green, 0; blue, 0 }  ][line width=0.08]  [draw opacity=0] (8.93,-4.29) -- (0,0) -- (8.93,4.29) -- cycle    ;
\draw    (219,146) -- (237.34,118.5) ;
\draw [shift={(239,116)}, rotate = 483.69] [fill={rgb, 255:red, 0; green, 0; blue, 0 }  ][line width=0.08]  [draw opacity=0] (8.93,-4.29) -- (0,0) -- (8.93,4.29) -- cycle    ;
\draw    (159,100) -- (68.55,156.41) ;
\draw [shift={(66,158)}, rotate = 328.05] [fill={rgb, 255:red, 0; green, 0; blue, 0 }  ][line width=0.08]  [draw opacity=0] (8.93,-4.29) -- (0,0) -- (8.93,4.29) -- cycle    ;
\draw    (57.1,95.8) -- (50.42,48.97) ;
\draw [shift={(50,46)}, rotate = 441.89] [fill={rgb, 255:red, 0; green, 0; blue, 0 }  ][line width=0.08]  [draw opacity=0] (8.93,-4.29) -- (0,0) -- (8.93,4.29) -- cycle    ;
\draw    (74.2,112.9) -- (157.96,202.8) ;
\draw [shift={(160,205)}, rotate = 227.03] [fill={rgb, 255:red, 0; green, 0; blue, 0 }  ][line width=0.08]  [draw opacity=0] (8.93,-4.29) -- (0,0) -- (8.93,4.29) -- cycle    ;
\draw   (252.8,177.1) .. controls (252.8,167.66) and (260.46,160) .. (269.9,160) .. controls (279.34,160) and (287,167.66) .. (287,177.1) .. controls (287,186.54) and (279.34,194.2) .. (269.9,194.2) .. controls (260.46,194.2) and (252.8,186.54) .. (252.8,177.1) -- cycle ;
\draw    (269,160) -- (259.65,118.13) ;
\draw [shift={(259,115.2)}, rotate = 437.42] [fill={rgb, 255:red, 0; green, 0; blue, 0 }  ][line width=0.08]  [draw opacity=0] (8.93,-4.29) -- (0,0) -- (8.93,4.29) -- cycle    ;
\draw   (430,62.9) .. controls (430,53.46) and (437.66,45.8) .. (447.1,45.8) .. controls (456.54,45.8) and (464.2,53.46) .. (464.2,62.9) .. controls (464.2,72.34) and (456.54,80) .. (447.1,80) .. controls (437.66,80) and (430,72.34) .. (430,62.9) -- cycle ;
\draw   (356,178) -- (396,178) -- (396,203.2) -- (356,203.2) -- cycle ;
\draw   (468,116.9) .. controls (468,107.46) and (475.66,99.8) .. (485.1,99.8) .. controls (494.54,99.8) and (502.2,107.46) .. (502.2,116.9) .. controls (502.2,126.34) and (494.54,134) .. (485.1,134) .. controls (475.66,134) and (468,126.34) .. (468,116.9) -- cycle ;
\draw   (370,132.9) .. controls (370,123.46) and (377.66,115.8) .. (387.1,115.8) .. controls (396.54,115.8) and (404.2,123.46) .. (404.2,132.9) .. controls (404.2,142.34) and (396.54,150) .. (387.1,150) .. controls (377.66,150) and (370,142.34) .. (370,132.9) -- cycle ;
\draw   (528,114) -- (568,114) -- (568,139.2) -- (528,139.2) -- cycle ;
\draw   (350,40) -- (390,40) -- (390,65.2) -- (350,65.2) -- cycle ;
\draw    (376,178) -- (385.99,152.79) ;
\draw [shift={(387.1,150)}, rotate = 471.62] [fill={rgb, 255:red, 0; green, 0; blue, 0 }  ][line width=0.08]  [draw opacity=0] (8.93,-4.29) -- (0,0) -- (8.93,4.29) -- cycle    ;
\draw    (390,60) -- (427.01,62.68) ;
\draw [shift={(430,62.9)}, rotate = 184.15] [fill={rgb, 255:red, 0; green, 0; blue, 0 }  ][line width=0.08]  [draw opacity=0] (8.93,-4.29) -- (0,0) -- (8.93,4.29) -- cycle    ;
\draw    (528,124) -- (505.09,117.7) ;
\draw [shift={(502.2,116.9)}, rotate = 375.39] [fill={rgb, 255:red, 0; green, 0; blue, 0 }  ][line width=0.08]  [draw opacity=0] (8.93,-4.29) -- (0,0) -- (8.93,4.29) -- cycle    ;
\draw    (464.2,62.9) .. controls (496.83,56.53) and (526.3,75.42) .. (546.76,111.75) ;
\draw [shift={(548,114)}, rotate = 241.53] [fill={rgb, 255:red, 0; green, 0; blue, 0 }  ][line width=0.08]  [draw opacity=0] (8.93,-4.29) -- (0,0) -- (8.93,4.29) -- cycle    ;
\draw    (468,116.9) -- (398.29,176.06) ;
\draw [shift={(396,178)}, rotate = 319.68] [fill={rgb, 255:red, 0; green, 0; blue, 0 }  ][line width=0.08]  [draw opacity=0] (8.93,-4.29) -- (0,0) -- (8.93,4.29) -- cycle    ;
\draw    (387.1,115.8) -- (380.42,68.97) ;
\draw [shift={(380,66)}, rotate = 441.89] [fill={rgb, 255:red, 0; green, 0; blue, 0 }  ][line width=0.08]  [draw opacity=0] (8.93,-4.29) -- (0,0) -- (8.93,4.29) -- cycle    ;
\draw   (280,107.5) -- (322,107.5) -- (322,100) -- (350,115) -- (322,130) -- (322,122.5) -- (280,122.5) -- cycle ;

\draw (108,34.4) node [anchor=north west][inner sep=0.75pt]    {$a_{1}$};
\draw (167,84.4) node [anchor=north west][inner sep=0.75pt]    {$a_{2}$};
\draw (48,104.4) node [anchor=north west][inner sep=0.75pt]    {$a_{3}$};
\draw (204,152.4) node [anchor=north west][inner sep=0.75pt]    {$a_{4}$};
\draw (36,162.4) node [anchor=north west][inner sep=0.75pt]    {$b_{3}$};
\draw (229,94.4) node [anchor=north west][inner sep=0.75pt]    {$b_{2}$};
\draw (90,214.4) node [anchor=north west][inner sep=0.75pt]    {$b_{4}$};
\draw (31,24.4) node [anchor=north west][inner sep=0.75pt]    {$b_{1}$};
\draw (170,208.4) node [anchor=north west][inner sep=0.75pt]    {$b_{5}$};
\draw (261,168.4) node [anchor=north west][inner sep=0.75pt]    {$a_{5}$};
\draw (438,54.4) node [anchor=north west][inner sep=0.75pt]    {$a_{1}$};
\draw (476,108.4) node [anchor=north west][inner sep=0.75pt]    {$a_{2}$};
\draw (378,124.4) node [anchor=north west][inner sep=0.75pt]    {$a_{3}$};
\draw (366,182.4) node [anchor=north west][inner sep=0.75pt]    {$b_{3}$};
\draw (538,118.4) node [anchor=north west][inner sep=0.75pt]    {$b_{2}$};
\draw (361,44.4) node [anchor=north west][inner sep=0.75pt]    {$b_{1}$};

\end{tikzpicture}
}
\end{center}
\caption{Suppose we run the kernel on an instance with the agent set $A = \lbrace a_{1},a_{2},a_{3},a_{4},a_{5} \rbrace$, the item set $I = \lbrace b_{1},b_{2},b_{3},b_{4},b_{5} \rbrace$, in which $p(a_{1}) = b_{1}$, $p(a_{2})=b_{2}$, $p(a_{3}) = b_{3}$, $p(a_{4}) = b_{\emptyset}$ and $p(a_{5}) = b_{\emptyset}$. The kernel removes $a_{4}$, $a_{5}$, $b_{4}$, and $b_{5}$ since they are not matched by the assignment. The resulting instance cannot contain self loops but only trading cycles as shown in the figure.}\label{KernelFigure}
\end{figure*}

\begin{theorem}\label{theorem:kernelAlloc}
The problems \vAAssign, \vDAssign, and \vCAssign\ admit kernels of size $\ell \cdot (\nalloc)^{2}$.
\end{theorem}
\begin{proof}
\begingroup
\removelatexerror
\begin{algorithm}[t]
\setstretch{1.2}
\DontPrintSemicolon
\BlankLine
        \If{$k = 1$ or the input is an instance of \vDAssign}{
	        $ fewLoops \longleftarrow preprocessing()$\;
	        \If {$ fewLoops = \false$}{
		        \KwRet{\no}\;
	        }
	    }
	
		$A_{\emptyset} \longleftarrow \lbrace a \in A\  |\  p(a) = b_{\emptyset} \rbrace$\;
		$I_{\emptyset} \longleftarrow \lbrace b \in I\  |\  \text{there does not exist $a \in A$ such that $p(a)=b$} \rbrace$\;
		$A \longleftarrow A \setminus A_{\emptyset}$\;
		$I \longleftarrow I \setminus I_{\emptyset}$\; 
		\ForEach{$i \in [\ell]$}{
			remove from $P_{i}$ the preference lists of the agents in $A_{\emptyset}$\;
			remove from the preference lists in $P_{i}$ all the items from $I_{\emptyset}$\;
			remove from $p$ the entries that correspond to the agents from $A_{\emptyset}$\;
		}	
	
	\KwRet{the reduced instance}\;
 \caption{Kernel for \vAAssign, \vDAssign\ and \vCAssign\ of size $\ell \cdot (\nalloc)^{2}$.}
 \label{algorithm:kernelAC}
\end{algorithm}
\endgroup

\begingroup
\removelatexerror
\begin{algorithm}[t]
\SetKwInOut{Input}{input}\SetKwInOut{Output}{output}
\Input{An instance $(A,I,P_{1},\ldots,P_{\ell},k,\alpha,p)$}
\Output{Does every agent $a \in A$ have at least $\alpha$ layers where it does not admit self loops?}
\setstretch{1.2}
\DontPrintSemicolon
\BlankLine
\ForEach{$a \in A$}{
	\If{$a$ admits self loops in at least $\ell - \alpha +1$ layers}{
		\KwRet{$\false$}\;
	}
}
\KwRet{$\true$}\;
\caption{Preprocessing step for the kernel for \vAAssign, \vDAssign, and \vCAssign.}
 \label{algorithm:preprocessing}
\end{algorithm}
\endgroup

We provide the same kernel for all problems. The kernel relies on the fact that only agents assigned to items different than $b_{\emptyset}$ can admit trading cycles. Thus, if $k=1$, or the input is an instance of \vDAssign, the kernel begins with a preprocessing step, in which it verifies whether each agent does not admit self loops in at least $\alpha$ layers in polynomial time. If it finds that the given instance is a \no-instance\ after the preprocessing step, it returns \no. After that (for any problem and for any value of $k$), the kernel removes from the instance (i) all the agents assigned to $b_{\emptyset}$ together with their preference lists, and (ii) all the items which are not allocated to any agent and their appearances in the preferences of the agents. In other words, the kernel first checks whether the optimality is violated due to the optimality requirement regarding the self loops (possible only when $k=1$ or the instance is of \vDAssign). After that, it keeps only the agents and the items that are matched by the assignment $p$ since they are the only candidates that can appear in trading cycles. The resulting trading graphs cannot contain self loops but only trading cycles (see \Ficref{KernelFigure}). The kernel is given in \Alcref{algorithm:kernelAC} and the preprocessing step is given in \Alcref{algorithm:preprocessing}.

The preprocessing step can be implemented in polynomial time by constructing the trading graphs of the layers with respect to $p$ and considering for each $a \in A$ the self loops it admits in each trading graph. The kernel can clearly be implemented in polynomial time as well. The correctness of \Alcref{algorithm:preprocessing} is due to \Obcref{observation:notAOptimalA}. Suppose that we run the kernel on an instance $D_{1} = (A,I,P_{1},\ldots,P_{\ell},\alpha,k,p)$ of \vAAssign, \vDAssign, or \vCAssign, and the kernel outputs the instance $D_{2} = (A',I',P_{1}',\ldots,P_{\ell}',\alpha,k,p')$. We claim the following.

\begin{myclaim}\label{claim:kernelClaim1}
Let $i \in [\ell]$, and let $\mathcal{C}_{i}$ be the set of trading cycles in the trading graph of $P_{i}$ with respect to $p$. Let $\mathcal{C}_{i}'$ be the set of trading cycles in the trading graph of $P_{i}'$ with respect to $p'$. Then $\mathcal{C}_{i} = \mathcal{C}_{i}'$.
\end{myclaim}
\begin{proof}

\smallskip
\noindent
$\subseteq$: Let $C = (a_{1},b_{1},\ldots,a_{t},b_{t}) \in \mathcal{C}_{i}$ be a trading cycle in the trading graph of $P_{i}$ with respect to $p$. Since for each $j \in [t]$, $a_{j}$ is assigned to an item, and $b_{j}$ is assigned to an agent, we have that $a_{j} \in A \setminus A_{\emptyset}$ and $b_{j} \in I \setminus I_{\emptyset}$. Thus, the preferences of $a_{j}$ on the items from $I \setminus I_{\emptyset}$ in layer $i$ remain unchanged. Moreover, we have that  $p'(a_{j}) = p(a_{j})$. Then, $C$ is a trading cycle in the trading graph of $P_{i}'$ with respect to $p'$.

\smallskip
\noindent
$\supseteq$: Let  $C = (a_{1},b_{1},\ldots,a_{t},b_{t}) \in \mathcal{C}_{i}'$ be a trading cycle in the trading graph of $P_{i}'$ with respect to $p'$. The kernelization algorithm keeps only the agents assigned to items, and their assigned items. Thus, for each $i \in [t]$, $a_{i}$ and $b_{i}$ are assigned by $p$, i.e.~$p' \subseteq p$. We have that the trading graph of $P_{i}'$ with respect to $p'$ is a sub-graph of the trading graph of $P_{i}$ with respect to $p$. Hence, $C$ is a trading cycle in layer $i$ of $D_{1}$ with respect to $p$.\qed
\end{proof}

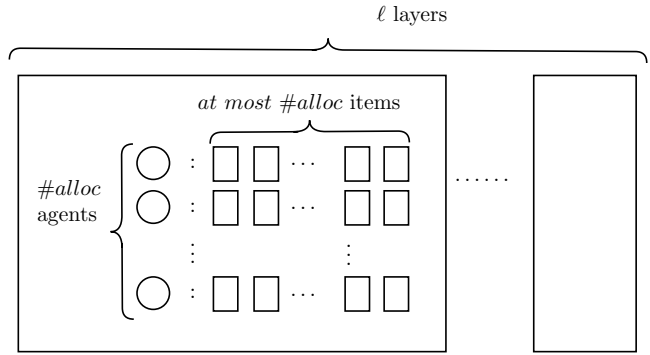
\begin{figure}[t]
\center \scalebox{0.82}{
\tikzset{every picture/.style={line width=0.75pt}} 

\begin{tikzpicture}[x=0.75pt,y=0.75pt,yscale=-1,xscale=1]

\draw   (10,50) -- (273,50) -- (273,220) -- (10,220) -- cycle ;
\draw   (83,104) .. controls (83,98.48) and (87.48,94) .. (93,94) .. controls (98.52,94) and (103,98.48) .. (103,104) .. controls (103,109.52) and (98.52,114) .. (93,114) .. controls (87.48,114) and (83,109.52) .. (83,104) -- cycle ;
\draw   (130,94) -- (145,94) -- (145,115) -- (130,115) -- cycle ;
\draw   (155,94) -- (170,94) -- (170,115) -- (155,115) -- cycle ;
\draw   (211,94) -- (226,94) -- (226,115) -- (211,115) -- cycle ;
\draw   (235,94) -- (250,94) -- (250,115) -- (235,115) -- cycle ;
\draw   (83,131) .. controls (83,125.48) and (87.48,121) .. (93,121) .. controls (98.52,121) and (103,125.48) .. (103,131) .. controls (103,136.52) and (98.52,141) .. (93,141) .. controls (87.48,141) and (83,136.52) .. (83,131) -- cycle ;
\draw   (130,121) -- (145,121) -- (145,142) -- (130,142) -- cycle ;
\draw   (155,121) -- (170,121) -- (170,142) -- (155,142) -- cycle ;
\draw   (211,121) -- (226,121) -- (226,142) -- (211,142) -- cycle ;
\draw   (235,121) -- (250,121) -- (250,142) -- (235,142) -- cycle ;
\draw   (83,184) .. controls (83,178.48) and (87.48,174) .. (93,174) .. controls (98.52,174) and (103,178.48) .. (103,184) .. controls (103,189.52) and (98.52,194) .. (93,194) .. controls (87.48,194) and (83,189.52) .. (83,184) -- cycle ;
\draw   (130,174) -- (145,174) -- (145,195) -- (130,195) -- cycle ;
\draw   (155,174) -- (170,174) -- (170,195) -- (155,195) -- cycle ;
\draw   (211,174) -- (226,174) -- (226,195) -- (211,195) -- cycle ;
\draw   (235,174) -- (250,174) -- (250,195) -- (235,195) -- cycle ;
\draw   (250.5,91.2) .. controls (250.47,86.53) and (248.12,84.22) .. (243.45,84.25) -- (199.2,84.53) .. controls (192.53,84.58) and (189.19,82.27) .. (189.16,77.6) .. controls (189.19,82.27) and (185.87,84.62) .. (179.2,84.66)(182.2,84.64) -- (134.95,84.95) .. controls (130.28,84.98) and (127.97,87.32) .. (128,91.99) ;
\draw   (80.5,92.2) .. controls (75.83,92.22) and (73.51,94.56) .. (73.53,99.23) -- (73.71,136.13) .. controls (73.74,142.8) and (71.42,146.14) .. (66.75,146.17) .. controls (71.42,146.14) and (73.77,149.46) .. (73.8,156.13)(73.79,153.13) -- (73.98,193.03) .. controls (74,197.7) and (76.34,200.02) .. (81.01,200) ;
\draw   (327,50) -- (390,50) -- (390,220) -- (327,220) -- cycle ;
\draw   (396.5,42.33) .. controls (396.53,37.66) and (394.21,35.32) .. (389.54,35.3) -- (210.54,34.33) .. controls (203.87,34.29) and (200.55,31.94) .. (200.58,27.27) .. controls (200.55,31.94) and (197.21,34.25) .. (190.54,34.22)(193.54,34.24) -- (11.54,33.25) .. controls (6.87,33.22) and (4.53,35.54) .. (4.5,40.21) ;

\draw (114,100) node [anchor=north west][inner sep=0.75pt]   [align=left] {:};
\draw (176,102) node [anchor=north west][inner sep=0.75pt]   [align=left] {$\displaystyle \dotsc $};
\draw (114,127) node [anchor=north west][inner sep=0.75pt]   [align=left] {:};
\draw (176,129) node [anchor=north west][inner sep=0.75pt]   [align=left] {$\displaystyle \dotsc $};
\draw (114,144) node [anchor=north west][inner sep=0.75pt]   [align=left] {$\displaystyle \vdots $};
\draw (210,145) node [anchor=north west][inner sep=0.75pt]   [align=left] {$\displaystyle \vdots $};
\draw (114,180) node [anchor=north west][inner sep=0.75pt]   [align=left] {:};
\draw (176,183) node [anchor=north west][inner sep=0.75pt]   [align=left] {$\displaystyle \dotsc $};
\draw (119,60) node [anchor=north west][inner sep=0.75pt]   [align=left] {$\displaystyle at\ most\ \#alloc$ items};
\draw (20,106) node [anchor=north west][inner sep=0.75pt]   [align=left] { \\$\displaystyle \#alloc$ \\agents};
\draw (277,112) node [anchor=north west][inner sep=0.75pt]   [align=left] {$\displaystyle \dotsc \dotsc $};
\draw (229,6) node [anchor=north west][inner sep=0.75pt]   [align=left] {$\displaystyle \ell $ layers};

\end{tikzpicture}
}
\caption{The kernel removes all the items which are not allocated and all the agents which are not assigned to items. Thus, its total size is $\OO(\ell \cdot (\nalloc)^{2})$.}\label{KernelSize}
\end{figure}

By the correctness of the preprocessing step, and since the set of trading cycles remains the same in each layer by \Clcref{claim:kernelClaim1}, we have that the resulting instance of the kernel is equivalent to the input instance. Since the kernel keeps only the agents and the items which are matched by the assignment, we have that its size is $\OO(\ell \cdot (\nalloc)^{2})$ (see \Ficref{KernelSize}). An example containing a trading graph of a layer before and after executing the kernel is given in  \Ficref{KernelFigure}.\qed

\end{proof}

Since $\nalloc \leq \min \lbrace n,m \rbrace$, we conclude the following.
\begin{corollary}
The problems \vAAssign, \vDAssign, and \vCAssign\ admit polynomial kernels with respect to $\nalloc + \ell$, $n + \ell$ and $m + \ell$.
\end{corollary}

\begin{theorem}\label{theorem:fptAlg}
The problems \vAAssign, \vDAssign, and \vCAssign\ are solvable in time $\OO^{*}(2^{\nalloc})$.
\end{theorem}
\begin{proof}
We provide a dynamic programming algorithm for each of the problems. The algorithms perform an adaptation of the technique by Björklund et al.~\cite{10.1145/1250790.1250801} (to compute the Fast zeta and Möbius transform). Each algorithm begins by running the kernelization algorithm (\Alcref{algorithm:kernelAC}) to reduce the number of the agents to $\nalloc$. Then, the algorithms construct $\ell$ tables with boolean values, each containing $(\nalloc)^{2}\cdot 2^{\nalloc}$ entries. Before we describe the algorithms, let us define the following, which is required to understand the purpose of each entry.

\begin{definition}
Let $(A,I,P_{1},\ldots,P_{\ell},\alpha,k,p)$ be an instance of \vAAssign, \vDAssign\ or \vCAssign, and let $s,t \in A$. We say that there is a {\em trading path} from $s$ to $t$ in layer $i$ if there exist agents $a_{1},\ldots,a_{r} \in A$ such that the trading graph of $P_{i}$ with respect to $p$ contains the path $p(s) \rightarrow s \rightarrow p(a_{1}) \rightarrow a_{1} \rightarrow \ldots \rightarrow p(a_{r}) \rightarrow a_{r} \rightarrow p(t) \rightarrow t$.
\end{definition} 

Notice that a trading cycle is a trading path from an agent to itself. For an agent $s \in A$ and layer $i$, we denote
\[ N_{i}(s) = \lbrace a \in A | (s,p(a)) \text{ is an edge in the trading graph of $P_{i}$} \rbrace. \] That is, $N_{i}(s)$ is the set of agents whose items are preferred by $s$ in layer $i$ over its assigned item $p(s)$. Notice that the trading graph of $P_{i}$ with respect to $p$ contains the path $p(s) \rightarrow s \rightarrow p(a) \rightarrow a$ for each $a \in N_{i}(s)$. We now describe each algorithm separately.

\smallskip
\noindent\textbf{\vAAssign\ and \vDAssign.}
Given an instance $(A,I,P_{1},\ldots,P_{\ell},\alpha,k,p)$ of \vAAssign\ or \vDAssign, the algorithm first performs the kernelization algorithm defined in \Alcref{algorithm:kernelAC}, and updates the input instance accordingly to have size $\OO(\ell (\nalloc)^{2})$. If the kernelization algorithm returns \no, then the algorithm returns $\false$ because, by the correctness of the kernel, we are dealing with a \no-instance. Since the kernel reduces the number of agents to $\nalloc$, we assume that $|A| = \nalloc$. The algorithm initializes $\ell$ tables, $M_{1},\ldots,M_{\ell}$, whose entries will store boolean values defined as follows. For each $i \in [\ell]$, agents $s,t \in A$ and a subset of agents $X \subseteq A$: $M_{i}[s,t,X] = \true$ if there exists a trading path from $s$ to $t$ that contains only the agents from $X \cup \lbrace s, t \rbrace$ and their assigned items; $M_{i}[s,t,X] = \false$ otherwise. Notice that:

\noindent\scalebox{0.9}{
$
    M_{i}[s,t,\emptyset]= 
\begin{dcases}
    \true & \text{if } (s,p(t)) \text{ is an edge in the trading graph of $P_{i}$ w.r.t.~$p$}\\
    \false               & \text{otherwise}
\end{dcases}
$
}
and for each $X \subseteq A$ such that $X \neq \emptyset$:

\[
    M_{i}[s,t,X]= \bigvee\limits_{s' \in N_{i}(s)}{M_{i}[s',t,X \setminus \lbrace s' \rbrace]} 
\]

This is because each trading path from $s$ to $t$ must start with an item of some agent from the set $N_{i}(s)$. Each table $M_{i}$ contains $(\nalloc)^{2}\cdot 2^{\nalloc}$ entries and can be constructed in the same running time. Thus, the algorithm constructs the $\ell$ tables $M_{1},\ldots,M_{\ell}$ in total time of $\OO(\ell \cdot (\nalloc)^{2} \cdot 2^{\nalloc})$. Then, it verifies if for each subset of agents $K \subseteq A$ such that $|K| = k$ (or $|K| \leq k$ if the input is an instance of \vDAssign), there exist $\alpha$ layers where the agents in $K$ do not admit trading cycles. To perform this, for each such subset $K$, the algorithm picks a random agent $a \in K$ and checks whether at least $\alpha$ values among $M_{1}[a,a,K\setminus \lbrace a \rbrace],\ldots,M_{\ell}[a,a,K\setminus \lbrace a \rbrace]$ equal $\true$ (these values determine whether the agents from $K$ admit trading cycles, without additional agents). If every such subset $K$ satisfies this condition, then the algorithm returns $\true$. Otherwise, it returns $\false$. The total running time of the algorithm is $\OO^{*}(2^{\nalloc})$.

\smallskip
\noindent\textbf{\vCAssign.}
Here, given an instance $(A,I,P_{1},\ldots,P_{\ell},\alpha,k,p)$ of \vCAssign, we first run the algorithm for \vAAssign. If it returns $\false$, then we return $\false$ as well. Otherwise, we have the $\ell$ tables $M_{1},\ldots,M_{\ell}$ constructed by the algorithm for \vAAssign. We need to verify that for each $K \subseteq A$ such that $|K| = k$, the agents in $K$ do not admit trading cycles together with possibly other agents in at least $\alpha$ layers. Recall that the tables $M_{1},\ldots,M_{\ell}$ already verify whether each such subset $K$ does not admit trading cycles without additional agents. In order to adapt to the definition of $(k,\alpha)$-\coptimality, we define $\ell$ new tables $N_{1},\ldots,N_{\ell}$. For each $i \in [\ell]$, agents $s,t \in A$, and a subset $X \subseteq A$, $N_{i}[s,t,X] = \true$ if there exists a trading path in layer $i$ from $s$ to $t$ containing all the agents from $X \cup \lbrace s,t \rbrace$ with their assigned items and possibly additional agents with their assigned items; $N_{i}[s,t,X] = \false$ otherwise. Notice that for each $i \in [\ell]$: \scalebox{0.9}{
$
N_{i}[s,t,X] = \bigvee\limits_{Y\subseteq A \text{ s.t.~} X\subseteq Y}{M_{i}[s,t,Y]} = \bigvee\limits_{Z\subseteq A \setminus X}{M_{i}[s,t,X \cup Z]} 
$
}. To construct the tables $N_{1},\ldots,N_{\ell}$ efficiently - in time $\OO^{*}(2^{\nalloc})$ rather than $\OO^{*}(3^{\nalloc})$ as implied from the above equality, we perform the following. Assume that $A = \lbrace a_{1},\ldots,a_{r} \rbrace$ where $r = \nalloc$. We represent each subset of agents $X \subseteq A$ as a vector of $r$ bits $(x_{1},\ldots,x_{r})$ such that for each $i \in [r]$, $x_{i} = 1$ if $a_{i} \in X$, and $x_{i} = 0$ otherwise. Observe that for two subsets $X = (x_{1},\ldots,x_{r})$ and $Y = (y_{1},\ldots,y_{r})$, $X \subseteq Y$ if and only if for each $i \in [r]$, $x_{i} \leq y_{i}$. Thus, for every subset $X = (x_{1},\ldots,x_{r}) \subseteq A$, and agents $s,t \in A$, we have that:

\noindent\scalebox{0.86}{
$
N_{i}[s,t,(x_{1},\ldots,x_{r})] = \hspace{-1.5em}\bigvee\limits_{y_{1},\ldots,y_{r} \in \lbrace 0,1 \rbrace}{\hspace{-1.5em}\delta[x_{1} \leq y_{1} \wedge \ldots \wedge x_{r} \leq y_{r}]} \cdot M_{i}[s,t,(y_{1},\ldots,y_{r})]
$
}
where $\delta$ is an indicator that equals $1$ if the expression in it is true, and $0$ otherwise. For each $j \in [r]$, we define:

$  N_{i}^{(j)}[s,t,(x_{1},\ldots,x_{r})] = \bigvee\limits_{y_{1},\ldots,y_{j} \in \lbrace 0,1 \rbrace}{\delta[x_{1} \leq y_{1} \wedge \ldots \wedge x_{j} \leq y_{j}]} \cdot M_{i}[s,t,(y_{1},\ldots,y_{j}, x_{j+1},\ldots,x_{r})]
$. We also define $
N_{i}^{(0)}[s,t,(x_{1},\ldots,x_{r})] =  M_{i}[s,t,(x_{1},\ldots,x_{r})]
$.

Notice that $N_{i}[s,t,(x_{1},\ldots,x_{r})] = N_{i}^{(t)}[s,t,(x_{1},\ldots,x_{r})]$. We can compute the values $N_{i}^{(j)}[s,t,(x_{1},\ldots,x_{r})]$ efficiently by the following observation:
\vspace{1em}

\noindent\scalebox{0.81}{
$
    N_{i}^{(j)}[s,t,(x_{1},\ldots,x_{r})]= 
\begin{dcases}
    N_{i}^{(j-1)}[s,t,(x_{1},\ldots,x_{r})] & \hspace{-1.5em}\text{if } x_{j} = 1\\
    N_{i}^{(j-1)}[s,t,(x_{1},\ldots,x_{j-1},1,x_{j+1},\ldots,x_{r})]\  \vee & \hspace{-1.5em}\text{if } x_{j} = 0 \\
    \ \ \ \ \ \ \ \ \ \ \ N_{i}^{(j-1)}[s,t,(x_{1},\ldots,x_{j-1},0,x_{j+1},\ldots,x_{r})] 
\end{dcases}
$

}
\vspace{1em}

Each table $N_{i}^{(j)}$ can be constructed in time $\OO(2^{\nalloc+2})$, thus constructing all the tables will take time $\OO(\ell \cdot \nalloc \cdot 2^{\nalloc}) = \OO^{*}(2^{\nalloc})$.
Then, for each subset $K \subseteq A$ with $|K| = k$, the algorithm picks an arbitrary agent $a \in K$ and checks whether there exist $\alpha$ layers $i_{1},\ldots,i_{\alpha}$ such that $N_{i_{j}}[s,t,(x_{1},\ldots,x_{r})] = \false$ for each $j \in [\alpha]$. If this is the case, it returns $\true$, otherwise, it returns $\false$. The correctness of the algorithm follows directly by the definition of $(k,\alpha)$-\coptimality.\qed 
\end{proof}

Since the parameter $\nalloc$ is smaller or equal than $\min\lbrace n,m \rbrace$, we conclude the following.
\begin{corollary}
The problems \vAAssign\ and \vCAssign\ are \FPT\ with respect to $\nalloc$, $n$ and $m$.
\end{corollary}

We prove now that \vAAssign\ and \vDAssign\ are \XP\ with respect to $k$. We remark that this statement can also follow as a corollary of \Thcref{theorem:fptDPlusK} by first using the kernel in \Thcref{theorem:kernelAlloc} to reduce $n$ to be exactly {\em equal} to $\nalloc$, which is lower-bounded by $d$. As a stand-alone proof is also simple, we opted for it.

\begin{theorem}\label{theorem:xpAlg}
The problems \vAAssign\ and \vDAssign\ are solvable in time $\OO^{*}(n^{\OO(k)})$.
\end{theorem}
\begin{proof}
We provide an $\OO^{*}(n^{\OO(k)})$-time algorithm. The algorithm relies on the algorithm by Bellman \cite{10.1145/321105.321111} for \hc\ in directed graphs with running time $\OO^{*}(2^{n})$, where $n$ is the number of vertices. Consider an instance $D = (A,I,P_{1},\ldots,P_{\ell},\alpha,k,p)$ of \vAAssign. For each $i \in [\ell]$, we denote by $G_{i}$ the trading graph of layer $i$ with respect to $p$. For any subset $K \subseteq A$, we denote by $G_{i}[K]$ the subgraph of $G_{i}$ that contains only the vertices that correspond to the agents in $K$ and their assigned items. We first claim the following.

\begin{myclaim}\label{claim:XP1}
Let $i \in [\ell]$ and let $K \subseteq A$ be a subset of agents such that $|K| \geq 2$. Then, the agents in $K$ admit a trading cycle in layer $i$ with respect to $p$ if and only if $G_{i}[K]$ admits a directed Hamiltonian cycle. 
\end{myclaim}
\begin{proof}
\smallskip
\noindent
$\Rightarrow$: Assume that the agents in $K$ admit a trading cycle in layer $i$ with respect to $p$. Then, $G_{i}$ contains a cycle over the agents in $K$ and their assigned items. This cycle is a Hamiltonian cycle in $G_{i}[K]$.

\smallskip
\noindent
$\Leftarrow$: Assume that $G_{i}[K]$ admits a directed Hamiltonian cycle. Since $G_{i}[K]$ is a subgraph of $G_{i}$, we have that the agents in $K$ admit a cycle in $G_{i}$ with their assigned items. This implies that $K$ admits a trading cycle in layer $i$.\qed
\end{proof}

\begin{myclaim}\label{claim:XP2}
Given a subset $K \subseteq A$ such that $|K| \geq 2$ and a layer $i \in [\ell]$, it is possible to determine whether the agents in $K$ admit a trading cycle in layer $i$ in time $\OO^{*}(2^{|K|})$.
\end{myclaim}
\begin{proof}
By \Clcref{claim:XP1}, we can construct the trading graph of layer $i$ with respect to $p$, $G_{i}$. Then, we run the algorithm of Bellman \cite{10.1145/321105.321111} on the subgraph $G[K]$ to check whether it admits a directed Hamiltonian cycle. If the answer is yes, this means that $K$ admits a trading cycle in layer $i$, thus we return $\true$. Otherwise, we return $\false$.\qed
\end{proof} 

Using these claims, we provide the algorithm for \vAAssign\ and \vDAssign. If $k=1$ or the instance is an instance of \vDAssign, then the first step of the algorithm is to verify whether for each $a \in A$, there exist at least $\alpha$ layers where it does not admit self loops. This can be easily done in polynomial time by constructing the trading graphs of all the layers, $G_{1},\ldots,G_{\ell}$, and checking for each of them whether $a$ is contained in a cycle with another item and no other vertex. If there exists some agent that appears in self loops in at least $\ell - \alpha + 1$ layers, we return $\false$. Otherwise, we need to verify that every subset of agents $K \subseteq A$ such that $|K| = k$ (or $|K| \leq k$ for \vDAssign) does not admit trading cycles in at least $\alpha$ layers. To verify this, for each such subset $K \subseteq A$, we check for each $i \in [\ell]$ whether the agents in $K$ admit a trading cycle in layer $i$ using \Clcref{claim:XP2} in time $\OO^{*}(2^{|K|}) = \OO^{*}(2^{k})$. If there exists a subset $K$ whose agents admit trading cycles in at least $\ell - \alpha + 1 $ layers, then by \Obcref{observation:notAOptimalA}, we return $\false$. Otherwise, we return $\true$. Notice that the number of subsets $K \subseteq A$ with $|K| \leq k$ is $\OO(n^{k})$. For each such subset, checking whether it admits trading cycles in more than $\ell - \alpha + 1$ layers takes time $\OO^{*}(\ell \cdot 2^{k})$. Thus, the total running time is $\OO^{*}(n^{k} \cdot \ell \cdot 2^{k}) = \OO^{*}(n^{\OO(k)})$.\qed
\end{proof}

\begin{corollary}
The problems \vAAssign\ and \vDAssign\ are \XP\ with respect to $k$.
\end{corollary}

\begingroup
\removelatexerror
\begin{algorithm}[ht!]
\setstretch{1.2}
\SetKwInOut{Input}{input}\SetKwInOut{Output}{output}
\Input{An instance $(A,I,P_{1},\ldots,P_{\ell},k,\alpha,p)$ of \vAAssign\ (or \vDAssign)}
\Output{Is $p$ $(k,\alpha)$-\aoptimal? (or $(k,\alpha)$-\doptimal)}
\DontPrintSemicolon
\BlankLine
	run the kernel of size $\ell (\nalloc)^{2}$ from \Thcref{theorem:kernelAlloc} on the input instance\;
	\If {$\text{the kernel returns \no}$}{
		\KwRet{\no}\;
	}
	\Else{
		\ForEach{$i \in [\ell]$}{
		    $G_{i} \leftarrow \text{the trading graph of $P_{i}$}$\;
		}
		\ForEach{$i \in [\ell]$}{
			\ForEach{$\text{trading cycle $C$ in $P_{i}$ with $k$ agents}$ (or at most $k$ agents)}{
				$A_{C} \leftarrow \text{the agents appearing in $C$}$\;
				\ForEach{$j \in [\ell] \setminus \lbrace i \rbrace$}{
				$G_{j}[A_{C}] \leftarrow \text{sub-graph of $G_{i}$ on $A_{C}$ and their items}$\;	
				}
				\If{$\text{at least $(\ell - \alpha)$ $j \in [\ell] \setminus \lbrace i \rbrace$ satisfy that $G_{j}[A_{C}]$}$ contains a Hamiltonian cycle}{
				\KwRet{\no}				
				}
			}	
	}
	}
	\KwRet{\yes}\;
 \caption{Algorithm for \vAAssign\ and \vDAssign\ with running time $\OO^{*}(d^{k})$.}
 \label{algorithm:FPT_K_Plus_D}
\end{algorithm}
\endgroup

We now turn to consider $d$ as a parameter, in addition to $k$.

\begin{theorem}\label{theorem:fptDPlusK}
\vAAssign\ and \vDAssign\ are solvable in time $\OO^{*}(d^{k})$.
\end{theorem}
\begin{proof}
We provide an algorithm for \vAAssign\ and \vDAssign\ (formally describe in \Alcref{algorithm:FPT_K_Plus_D}). The algorithm is based on the observation that since each agent prefers at most $d - 1$ items over its own assigned item (or at most $d$ if it not allocated an item), the number of possible trading cycles with exactly $k$ agents (or at most $k$ agents if the input is an instance of \vDAssign) is at most $\OO(n \cdot d ^{k})$ (where $n$ is the number of agents).

The algorithm first performs the kernelization algorithm in \Thcref{theorem:kernelAlloc} in order to test whether each agent does not admit self loops in at least $\ell - \alpha + 1$ layers (when $k=1$), and to reduce the instance size to $\OO(\ell \cdot (\nalloc)^{2})$. Then, for each $i \in [\ell]$, the algorithm considers all the trading cycles in layer $i$ with $k$ (or at most $k$) agents. For each such trading cycle $C$, it checks in which other layers the agents in $C$ admit trading cycles. If it finds out that there are at least $\ell - \alpha + 1$ layers in which the agents admit trading cycles, then by \Obcref{observation:notAOptimalA}, it returns \no. Otherwise, the algorithm returns $\yes$. In order to check in which layers the agents in $C$ admit trading cycles, we use a similar technique as in \Thcref{theorem:xpAlg}, which uses the $\OO^{*}(2^{n})$-time algorithm for \hc\ on directed graphs by Bellman \cite{10.1145/321105.321111}. For a subset $A_{C}$ of agents that appear in a trading cycle $C$ in layer $i$, the algorithm constructs for each $j \in [\ell] \setminus \lbrace i \rbrace$ the graph $G_{j}[A_{C}]$, which is the sub-graph of the trading graph in layer $j$ containing only the agents from $A_{C}$ and their assigned items as well as the edges between them. Using the algorithm of Bellman \cite{10.1145/321105.321111}, and based on the previous observations, the algorithm can be implemented in time $\OO^{*}(\ell \cdot n \cdot d^{k} \cdot 2^{k}) = \OO^{*}(d^{k})$. Its correctness is due to \Obcref{observation:notAOptimalA}, and the observation that $G_{j}$ contains a trading cycle over the agents set $A_{C}$ if and only if $G_{j}[A_{C}]$ admits a Hamiltonian cycle.\qed
\end{proof}

\begin{corollary}
\vAAssign\ and \vDAssign\ are \FPT\ with respect to $k + d$.
\end{corollary}

\begin{theorem}\label{theorem:coWOH}
The problems \vAAssign, \vDAssign\ and \vCAssign\ are \coWOH\ with respect to $k+\ell$. 
\end{theorem}
\begin{proof}
We provide a parameterized (and also polynomial) reduction from \MCIS\ (defined immediately) to \coVAAssign. We will later explain why it is also a reduction to \coVDAssign\ and \coVCAssign. The \MCIS\ problem was proved to be \WOH\ by Fellows et al.~\cite{fellows}. The input of \MCIS\ consists of an undirected graph $G = (V,E)$, an integer $2\leq \tildeK \leq |V|$, and a coloring $c : V \rightarrow [\tildeK]$ that colors the vertices in $G$ with $\tildeK$ colors. The task is to decide whether $G$ admits a {\em multicolored independent set} of size $\tildeK$, which is an independent set (i.e.~a vertex subset with pair-wise non-adjacent vertices) $V' \subseteq V$ that satisfies $\lbrace c(v') | v' \in V' \rbrace = [\tildeK]$ and $|V'| = \tildeK$.

Given an instance $(G = (V,E),\tildeK,c)$ of \MCIS, denote $V = \lbrace v_{1},\ldots,v_{n} \rbrace$. We construct an instance of \coVAAssign\ with $n$ agents, $n$ items, ${\tildeK \choose 2}$ layers, $\alpha = 1$, $k = \tildeK$, and an assignment $p$. We will prove that there exists a subset of agents of size $k$ that admits trading cycles in all the layers with respect to $p$ if and only if $G$ contains a multicolored independent set of size $\tildeK$. We first create an agent $a_{i}$ and an item $b_{i}$ for each vertex $v_{i} \in V$. Thus, the agent set and the item set of the constructed instance are $A = A_{n}$ and $I = I_{n}$, respectively. We also set $p = p_{n}$ (recall that $p(a_{i})=b_{i}$ for each $i \in [n]$). Intuitively, if there exists a multicolored independent set in $G$, then the agents corresponding to the vertices in the multicolored independent set will admit trading cycles in all the layers. Each layer corresponds to a pair of colors $\lbrace u , w \rbrace$, and ensures that (i) every trading cycle contains exactly $k$ agents that correspond to vertices colored by all the colors; and (ii) the vertices with the colors $u$ and $w$ whose agents appear in a trading cycle will not be adjacent in $G$.

For each $s \in [\tildeK]$, denote $A(s) = \lbrace a_{i} \in A | c(v_{i}) = s \rbrace$ and $I(s) = \lbrace b_{i} \in I | c(v_{i}) = s \rbrace$, namely, the agents and the items that correspond to vertices colored $s$ by the coloring $c$, respectively. Let $u,w \in [\tildeK]$ be two different colors such that $u < w$; denote the other colors by $s_{1},\ldots,s_{\tildeK - 2} \in [\tildeK] \setminus \lbrace u , w \rbrace$ such that $s_{1} < \ldots < s_{\tildeK - 2}$. The preference profile $P_{\lbrace u,w \rbrace}^{MCIS}$ is defined as follows.

\begin{framed}
\begin{itemize}
  \item $a_{r}\ :\ I(s_{j+1})\ >\ b_{r}\ \forall v_{r} \in V, j \in [\tildeK-3]$ s.t.~ $c(v_{r})=s_{j}$
  \item $a_{r}\ :\ I(u)\ >\ b_{r}\ \forall v_{r} \in V$ s.t.~ $c(v_{r})=s_{\tildeK-2}$
  \item $a_{r}\ :\ I(w) \cap \lbrace a_{i} | \lbrace  v_{i},v_{r} \rbrace \notin E \rbrace \ >\ b_{r}\ \forall v_{r} \in V$ s.t.~ $c(v_{r})=u$
  \item $a_{r}\ :\ I(s_{1})\ >\ b_{r}\ \forall v_{r} \in V$ s.t.~ $c(v_{r})=w$
\end{itemize}
\end{framed}

An example of a possible trading cycle in layer $P_{\lbrace u,w \rbrace}^{MCIS}$ is given in \Ficref{COWOHFigure3}.

Let us now claim the following, regarding the trading cycles in such layers.

\begin{figure*}[t]
\center \scalebox{0.9}{
\tikzset{every picture/.style={line width=0.75pt}} 

\begin{tikzpicture}[x=0.75pt,y=0.75pt,yscale=-1,xscale=1]

\draw  [fill={rgb, 255:red, 141; green, 255; blue, 17 }  ,fill opacity=1 ] (416.56,70.6) .. controls (416.56,61.54) and (423.94,54.2) .. (433.04,54.2) .. controls (442.13,54.2) and (449.51,61.54) .. (449.51,70.6) .. controls (449.51,79.66) and (442.13,87) .. (433.04,87) .. controls (423.94,87) and (416.56,79.66) .. (416.56,70.6) -- cycle ;
\draw  [fill={rgb, 255:red, 141; green, 255; blue, 17 }  ,fill opacity=1 ] (478.99,55.06) -- (511.5,55.06) -- (511.5,86.14) -- (478.99,86.14) -- cycle ;
\draw    (479.42,70.77) -- (452.51,70.62) ;
\draw [shift={(449.51,70.6)}, rotate = 360.33000000000004] [fill={rgb, 255:red, 0; green, 0; blue, 0 }  ][line width=0.08]  [draw opacity=0] (8.93,-4.29) -- (0,0) -- (8.93,4.29) -- cycle    ;
\draw    (28.47,88) -- (28.47,115) ;
\draw [shift={(28.47,118)}, rotate = 270] [fill={rgb, 255:red, 0; green, 0; blue, 0 }  ][line width=0.08]  [draw opacity=0] (8.93,-4.29) -- (0,0) -- (8.93,4.29) -- cycle    ;
\draw  [fill={rgb, 255:red, 255; green, 165; blue, 0 }  ,fill opacity=1 ] (290.85,70.6) .. controls (290.85,61.54) and (298.22,54.2) .. (307.32,54.2) .. controls (316.42,54.2) and (323.79,61.54) .. (323.79,70.6) .. controls (323.79,79.66) and (316.42,87) .. (307.32,87) .. controls (298.22,87) and (290.85,79.66) .. (290.85,70.6) -- cycle ;
\draw  [fill={rgb, 255:red, 255; green, 165; blue, 0 }  ,fill opacity=1 ] (353.27,55.06) -- (385.78,55.06) -- (385.78,86.14) -- (353.27,86.14) -- cycle ;
\draw    (353.71,70.77) -- (326.79,70.62) ;
\draw [shift={(323.79,70.6)}, rotate = 360.33000000000004] [fill={rgb, 255:red, 0; green, 0; blue, 0 }  ][line width=0.08]  [draw opacity=0] (8.93,-4.29) -- (0,0) -- (8.93,4.29) -- cycle    ;
\draw    (416.56,70.6) -- (389.65,70.44) ;
\draw [shift={(386.65,70.43)}, rotate = 360.33000000000004] [fill={rgb, 255:red, 0; green, 0; blue, 0 }  ][line width=0.08]  [draw opacity=0] (8.93,-4.29) -- (0,0) -- (8.93,4.29) -- cycle    ;
\draw  [fill={rgb, 255:red, 159; green, 159; blue, 159 }  ,fill opacity=1 ] (166,70.6) .. controls (166,61.54) and (173.38,54.2) .. (182.47,54.2) .. controls (191.57,54.2) and (198.95,61.54) .. (198.95,70.6) .. controls (198.95,79.66) and (191.57,87) .. (182.47,87) .. controls (173.38,87) and (166,79.66) .. (166,70.6) -- cycle ;
\draw  [fill={rgb, 255:red, 159; green, 159; blue, 159 }  ,fill opacity=1 ] (228.42,55.06) -- (260.94,55.06) -- (260.94,86.14) -- (228.42,86.14) -- cycle ;
\draw    (228.86,70.77) -- (201.95,70.62) ;
\draw [shift={(198.95,70.6)}, rotate = 360.33000000000004] [fill={rgb, 255:red, 0; green, 0; blue, 0 }  ][line width=0.08]  [draw opacity=0] (8.93,-4.29) -- (0,0) -- (8.93,4.29) -- cycle    ;
\draw    (291.72,70.6) -- (264.8,70.44) ;
\draw [shift={(261.8,70.43)}, rotate = 360.33000000000004] [fill={rgb, 255:red, 0; green, 0; blue, 0 }  ][line width=0.08]  [draw opacity=0] (8.93,-4.29) -- (0,0) -- (8.93,4.29) -- cycle    ;
\draw  [fill={rgb, 255:red, 255; green, 235; blue, 0 }  ,fill opacity=1 ] (12,71.6) .. controls (12,62.54) and (19.38,55.2) .. (28.47,55.2) .. controls (37.57,55.2) and (44.95,62.54) .. (44.95,71.6) .. controls (44.95,80.66) and (37.57,88) .. (28.47,88) .. controls (19.38,88) and (12,80.66) .. (12,71.6) -- cycle ;
\draw  [fill={rgb, 255:red, 255; green, 235; blue, 0 }  ,fill opacity=1 ] (74.42,56.06) -- (106.94,56.06) -- (106.94,87.14) -- (74.42,87.14) -- cycle ;
\draw    (74.86,71.77) -- (47.95,71.62) ;
\draw [shift={(44.95,71.6)}, rotate = 360.33000000000004] [fill={rgb, 255:red, 0; green, 0; blue, 0 }  ][line width=0.08]  [draw opacity=0] (8.93,-4.29) -- (0,0) -- (8.93,4.29) -- cycle    ;
\draw    (137.72,71.6) -- (110.8,71.44) ;
\draw [shift={(107.8,71.43)}, rotate = 360.33000000000004] [fill={rgb, 255:red, 0; green, 0; blue, 0 }  ][line width=0.08]  [draw opacity=0] (8.93,-4.29) -- (0,0) -- (8.93,4.29) -- cycle    ;
\draw   (518.5,50.2) .. controls (518.51,45.53) and (516.18,43.2) .. (511.51,43.19) -- (278.6,42.74) .. controls (271.93,42.73) and (268.6,40.39) .. (268.61,35.72) .. controls (268.6,40.39) and (265.27,42.71) .. (258.6,42.7)(261.6,42.71) -- (18.51,42.24) .. controls (13.84,42.23) and (11.51,44.56) .. (11.5,49.23) ;
\draw  [fill={rgb, 255:red, 255; green, 0; blue, 0 }  ,fill opacity=1 ] (77,136) .. controls (77,126.94) and (84.38,119.6) .. (93.47,119.6) .. controls (102.57,119.6) and (109.95,126.94) .. (109.95,136) .. controls (109.95,145.06) and (102.57,152.4) .. (93.47,152.4) .. controls (84.38,152.4) and (77,145.06) .. (77,136) -- cycle ;
\draw  [fill={rgb, 255:red, 255; green, 0; blue, 0 }  ,fill opacity=1 ] (13.49,120) -- (46,120) -- (46,151.07) -- (13.49,151.07) -- cycle ;
\draw    (47,133) -- (74,133) ;
\draw [shift={(77,133)}, rotate = 180] [fill={rgb, 255:red, 0; green, 0; blue, 0 }  ][line width=0.08]  [draw opacity=0] (8.93,-4.29) -- (0,0) -- (8.93,4.29) -- cycle    ;
\draw    (236,136.4) .. controls (302.17,106.75) and (358.92,229.1) .. (477.2,88.28) ;
\draw [shift={(478.99,86.14)}, rotate = 489.65] [fill={rgb, 255:red, 0; green, 0; blue, 0 }  ][line width=0.08]  [draw opacity=0] (8.93,-4.29) -- (0,0) -- (8.93,4.29) -- cycle    ;
\draw   (10,156.79) .. controls (10.01,161.46) and (12.35,163.78) .. (17.02,163.77) -- (113.77,163.43) .. controls (120.44,163.41) and (123.78,165.73) .. (123.8,170.4) .. controls (123.78,165.73) and (127.1,163.39) .. (133.77,163.36)(130.77,163.37) -- (230.52,163.02) .. controls (235.19,163.01) and (237.51,160.67) .. (237.5,156) ;
\draw  [fill={rgb, 255:red, 0; green, 0; blue, 255 }  ,fill opacity=1 ] (203.05,136.4) .. controls (203.05,127.34) and (210.43,120) .. (219.53,120) .. controls (228.62,120) and (236,127.34) .. (236,136.4) .. controls (236,145.46) and (228.62,152.8) .. (219.53,152.8) .. controls (210.43,152.8) and (203.05,145.46) .. (203.05,136.4) -- cycle ;
\draw  [fill={rgb, 255:red, 0; green, 0; blue, 255 }  ,fill opacity=1 ] (139.54,120.4) -- (172.05,120.4) -- (172.05,151.47) -- (139.54,151.47) -- cycle ;
\draw    (173.05,133.4) -- (200.05,133.4) ;
\draw [shift={(203.05,133.4)}, rotate = 180] [fill={rgb, 255:red, 0; green, 0; blue, 0 }  ][line width=0.08]  [draw opacity=0] (8.93,-4.29) -- (0,0) -- (8.93,4.29) -- cycle    ;
\draw    (110,133) -- (137,133) ;
\draw [shift={(140,133)}, rotate = 180] [fill={rgb, 255:red, 0; green, 0; blue, 0 }  ][line width=0.08]  [draw opacity=0] (8.93,-4.29) -- (0,0) -- (8.93,4.29) -- cycle    ;

\draw (142,69.4) node [anchor=north west][inner sep=0.75pt]    {$\dotsc $};
\draw (13,7) node [anchor=north west][inner sep=0.75pt]   [align=left] {Agents and items that correspond to $\displaystyle \tilde{k} -2$ vertices colored with $\displaystyle \left[\tilde{k}\right] \setminus \{u,w\}$};
\draw (11,172) node [anchor=north west][inner sep=0.75pt]   [align=left] {Agents and items that correspond to \\two non-adjacent vertices with colors $\displaystyle u$ and $\displaystyle w$};

\end{tikzpicture}
}
\caption{The form of possible trading cycles in $P_{\lbrace u,w \rbrace}^{MCIS}$.}\label{COWOHFigure3}
\end{figure*}
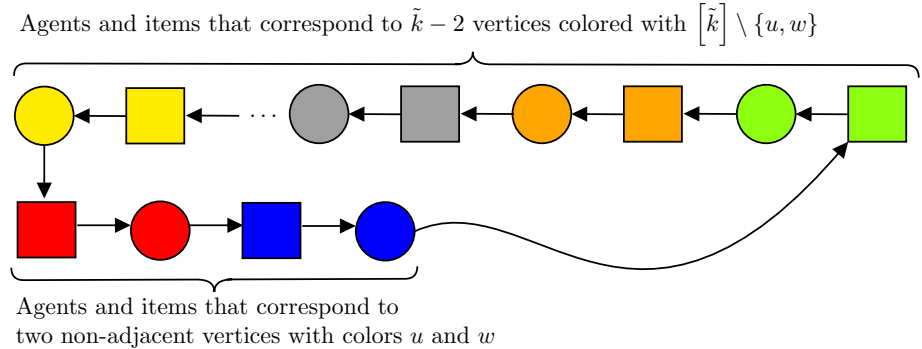

\begin{myclaim}
\label{WOHClaim}
A sequence $(a_{i_{1}},b_{i_{1}},\ldots,a_{i_{\tildeK}}, b_{i_{\tildeK}})$ (up to cyclic shifts) is a trading cycle in $P_{\lbrace u,w \rbrace}^{MCIS}$ if and only if:\begin{enumerate}
    \item $c(v_{i_{j}}) = s_{j}$ for each $j \in [\tildeK-2]$;
    \item $c(v_{i_{\tildeK - 1}}) = u$;
    \item $c(v_{i_{\tildeK}}) = w$;
    \item $\lbrace v_{i_{\tildeK - 1 }} , v_{i_{\tildeK}} \rbrace \notin E$.
\end{enumerate}
\end{myclaim}
\begin{proof}
By the construction of $P_{\lbrace u,w \rbrace}^{MCIS}$, notice that for each $j \in [\tildeK - 3]$, each agent $a_{r}$ such that $c(v_{r}) = s_{j}$ prefers all the items whose corresponding vertices are colored with $s_{j+1}$. Moreover, each agent $a_{r}$ whose corresponding vertex $v_{r}$ is colored with $u$ prefers all the items whose corresponding vertex is colored with $w$ and not adjacent to $v_{r}$ in $G$; and every $a_{r}$ whose corresponding vertex is colored with $w$ prefers all the items whose vertices are colored with $s_{1}$. Thus, every trading cycle must begin with a sequence $(a_{i_{1}},b_{i_{1}},\ldots,a_{i_{\tildeK-2}},b_{i_{\tildeK-2}})$, such that $a_{i_{j}}$ corresponds to a vertex colored with $s_{j}$ for each $j \in [\tildeK - 2]$; and after that, it contains a sequence $(a_{i_{\tildeK - 1}}, b_{i_{\tildeK - 1}},a_{i_{\tildeK}}, b_{i_{\tildeK}})$ such that $a_{i_{\tildeK - 1}}$ and $a_{i_{\tildeK}}$ correspond to non-adjacent vertices colored with $u$ and $w$, respectively. An illustration of the form of possible trading cycles in $P_{\lbrace u,w \rbrace}^{MCIS}$ is given in \Ficref{COWOHFigure3}.\qed  
\end{proof}

We prove now that the resulting instance is a \yesinstance\ of \coVAAssign\ if and only if $G$ admits a multicolored independent set. By \Obcref{observation:notAOptimalA}, we need to show that $G$ admits a multicolored independent set of size $k$ if and only if there exists a subset of agents of size $k$ that admits trading cycles in all the layers (since $\alpha = 1$).

\smallskip
\noindent
$\Rightarrow$: Assume that $G$ admits a multicolored independent set $V' = \lbrace v_{i_{1}},\ldots,v_{i_{\tildeK}} \rbrace$. Denote $K = \lbrace a_{i_{1}}, \ldots, a_{i_{\tildeK}} \rbrace$ (namely, the agents which correspond to the vertices in $V'$). Let $u,w \in [\tildeK]$ be two different colors such that $u < w$, and assume that $s_{1},\ldots,s_{\tildeK - 2} \in [\tildeK] \setminus \lbrace u,w \rbrace$ such that $s_{1} < \ldots < s_{\tildeK - 2}$. Suppose w.l.o.g.~that $c(v_{j}) = s_{j}$ for each $j \in [\tildeK-2]$, $c(v_{\tildeK - 1}) = u$, and $c(v_{\tildeK}) = w$ (since the vertices in $V'$ are colored with all the colors). Since $V'$ is an independent set, we have that $v_{\tildeK - 1}$ and $v_{\tildeK}$ are not adjacent in $G$. Thus, by \Clcref{WOHClaim}, $K$ admits a trading cycle in $P_{\lbrace u,w \rbrace}^{MCIS}$. Therefore, the resulting instance is a \yesinstance.

\smallskip
\noindent
$\Leftarrow$: Assume that there exists a subset $K = \lbrace a_{i_{1}},\ldots,a_{i_{\tildeK}} \rbrace \subseteq A_{n}$ of size $k = \tildeK$, which admits trading cycles in all the layers of the constructed instance. By \Clcref{WOHClaim}, $K$ contains agents which correspond to vertices colored with all the colors. Moreover, for each $u , w \in [\tildeK]$ such that $u < w$, $K$ contain two agents, $a_{i_{r}}$ and $a_{i_{t}}$, such that $c(v_{i_{r}}) = u$, $c(v_{i_{t}}) = w$, and $\lbrace v_{i_{r}} , v_{i_{t}} \rbrace \notin E$. This implies that the vertices in $V' = \lbrace v_{i_{1}},\ldots,v_{i_{\tildeK}} \rbrace$ are pair-wisely non-adjacent, and colored with all the colors. Thus, $V'$ is a multicolored independent set in $G$ of size $\tildeK$.

Since the parameter $k + \ell$ of the constructed instance depends only on the parameter $\tildeK$, we have that $\vAAssign$ is \coWOH\ with respect to $k+\ell$. Notice that \Clcref{WOHClaim} implies that all the possible trading cycles in each layer necessarily have exactly $\tildeK$ agents. Thus, by \Obcref{obs:eqAoptAndCopt} and \Obcref{lemma:aAndDEq}, the resulting instance can represent an equivalent instance of \coVCAssign\ or \coVDAssign. So, the results apply to \vCAssign\ and \vDAssign\ as well.\qed
\end{proof}
\section{\coNP-Hardness}
\label{sec:conphardness}
In this section, we prove that the problems \vAAssign, \vDAssign, and \vCAssign\ are \CONPH\ even when $k=n$, $\alpha = \ell = 1$ ($\ell = 2$ for \vDAssign), and $d = 3$. Before that, let us define two preference profiles that we will use in our next proofs. Let $G = (V,E)$ be a directed graph, and suppose that $V = \lbrace v_{1}, \ldots , v_{n} \rbrace$. For each vertex $v_{i}$ in $G$, we create one agent $a_{i}$ and one item $b_{i}$. We denote the agent set $A_{n} = \lbrace a_{1},\ldots,a_{n} \rbrace$, the item set $I_{n} = \lbrace b_{1},\ldots , b_{n} \rbrace$, and the assignment $p_{n}$ by $p_{n}(a_{i}) = b_{i}$ for each $i \in [n]$. We will construct the first preference profile, $P_{1}(G)$, over the agent set $A_{n}$ and the item set $I_{n}$ so that its trading graph with respect to $p_{n}$ will be derived from the graph $G$. Namely, if a subset of vertices $V' \subseteq V$ admits a directed cycle in $G$, then the corresponding agents and items of these vertices will admit a trading cycle in the trading graph of $P_{1}(G)$ with respect to $p_{n}$. Formally, we construct $P_{1}(G)$ as follows.
\begin{framed}[1\columnwidth]
\begin{itemize}
  \item $a_{i}\ $: $\lbrace b_{j} | (v_{i},v_{j}) \in E \rbrace\ \text{(in arbitrary order)}\ >\ b_{i}\ \ \forall i\in[n]$
\end{itemize}
\end{framed}

\begin{figure}[t]
\center \scalebox{0.7}{
\input{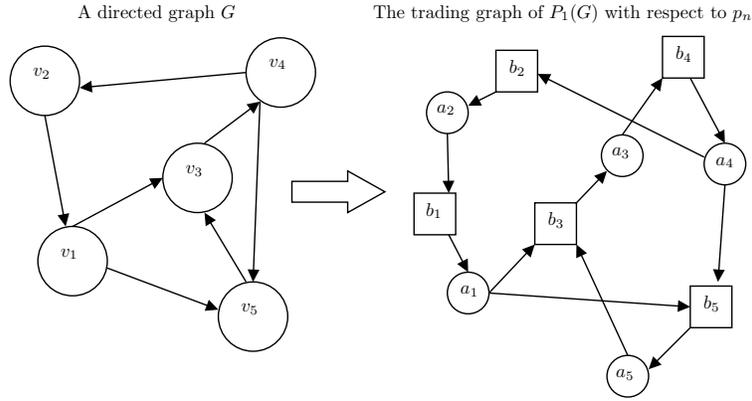}
}
\caption{$G$ is a directed graph with $n = 5$ vertices that contains a Hamiltonian cycle $v_{1} \rightarrow v_{5} \rightarrow v_{3} \rightarrow v_{4} \rightarrow v_{2} \rightarrow v_{1}$. Each vertex $v_{i}$ has an agent $a_{i}$ and an item $b_{i}$ such that $b_{i}$ is assigned to $a_{i}$ by the assignment $p_{n}$. Observe that the trading graph of $P_{1}(G)$ with respect to $p_{n}$ contains the trading cycle $(a_{1},b_{1},a_{5},b_{5},a_{3},b_{3},a_{4},b_{4},a_{2},b_{2})$, which corresponds to the aforementioned Hamiltonian cycle.}\label{P1G}
\end{figure}

An example of the construction of $P_{1}(G)$ is given in \Ficref{P1G}. Let us now prove the following lemma.

\begin{lemma}\label{lemma:coNPHLemma1}
A directed graph $G$ contains a directed cycle $(v_{i_{1}},\ldots, v_{i_{t}})$ if and only if $P_{1}(G)$ contains the trading cycle $(a_{i_{1}},b_{i_{1}},\ldots,a_{i_{t}},b_{i_{t}})$ with respect to $p_{n}$.
\end{lemma}
\begin{proof}
$\Rightarrow$: Suppose that $(v_{i_{1}},\ldots , v_{i_{t}})$ is a directed cycle in $G$. Note that there exists a directed edge from $v_{i_{j}}$ to $v_{i_{j+1}}$ for each $j \in [t-1]$ and from $v_{t}$ to $v_{1}$. Thus, by the construction of $P_{1}(G)$, $a_{i_{j}}$ prefers $b_{i_{j+1}}$ over its assigned item $p_{n}(a_{i_{j}}) = b_{i_{j}}$, and $a_{i_{t}}$ prefers $b_{i_{1}}$ over its assigned item $p_{n}(a_{i_{t}}) = b_{i_{t}}$. This yields the trading cycle $(a_{i_{1}},b_{i_{1}},\ldots,a_{i_{t}},b_{i_{t}})$.

\smallskip 
\noindent
$\Leftarrow$: Suppose that $(a_{i_{1}},b_{i_{1}}, \ldots ,a_{i_{t}},b_{i_{t}})$ is a trading cycle in $P_{1}(G)$ with respect to $p_{n}$. Then, for each $j \in [t-1]$, $a_{i_{j}}$ prefers $b_{i_{j+1}}$ over $b_{i_{j}}$, and $a_{i_{t}}$ prefers $b_{i_{1}}$ over $b_{i_{t}}$. By the construction of $P_{1}(G)$, $(v_{i_{j}},v_{i_{j+1}}) \in E$ for each $j \in [t-1]$, and $(v_{i_{t}},v_{i_{1}}) \in E$. This implies that $G$ contains the cycle $(v_{i_{1}},\ldots, v_{i_{t}})$.\qed
\end{proof}
 
Intuitively, if all the agents in $A_{n}$ admit a trading cycle in $P_{1}(G)$, then we can conclude that $G$ contains a directed cycle over all the vertices, namely, a Hamiltonian cycle, and vice versa.

\begin{corollary}
\label{lemma:coNPHLemma2}
The set $A_{n}$ admits a trading cycle in $P_{1}(G)$ with respect to $p_{n}$ if and only if $G$ contains a Hamiltonian Cycle. 
\end{corollary}
\begin{proof}
$\Rightarrow$: Suppose that $A_{n}$ admits a trading cycle in $P_{1}(G)$ with respect to $p_{n}$. Assume that this trading cycle is $(a_{i_{1}},b_{i_{1}},\ldots,a_{i_{n}},b_{i_{n}})$. By \Lecref{lemma:coNPHLemma1}, $G$ contains the cycle $(v_{i_{1}},\ldots,v_{i_{n}})$, which is a Hamiltonian cycle.

\smallskip
\noindent
$\Leftarrow$: Suppose that $G$ contains a Hamiltonian cycle $(v_{i_{1}},\ldots,v_{i_{n}})$. Then, by \Lecref{lemma:coNPHLemma1}, $P_{1}(G)$ contains the trading cycle $(a_{i_{1}},a_{i_{1}},\ldots,a_{i_{n}},b_{i_{n}})$ with respect to $p$. Observe that this is a trading cycle over all the agents.\qed
\end{proof}

For every $n$, the second preference profile, $P_{2}(n)$, will be used to enforce the size of subsets of agents admitting trading cycles in all the layers to be equal to exactly $n$. This will help us to prove hardness results for the \vDAssign\ problem.

\begin{figure}[t]
\center \scalebox{0.9}{
\tikzset{every picture/.style={line width=0.75pt}} 

\begin{tikzpicture}[x=0.75pt,y=0.75pt,yscale=-1,xscale=1]

\draw   (234.8,75.1) .. controls (234.8,64) and (243.8,55) .. (254.9,55) .. controls (266,55) and (275,64) .. (275,75.1) .. controls (275,86.2) and (266,95.2) .. (254.9,95.2) .. controls (243.8,95.2) and (234.8,86.2) .. (234.8,75.1) -- cycle ;
\draw   (151,28) -- (196.5,28) -- (196.5,59.2) -- (151,59.2) -- cycle ;
\draw   (226,139.8) -- (271.5,139.8) -- (271.5,171) -- (226,171) -- cycle ;
\draw   (164.89,202.3) .. controls (164.89,191.2) and (173.89,182.2) .. (184.99,182.2) .. controls (196.09,182.2) and (205.09,191.2) .. (205.09,202.3) .. controls (205.09,213.4) and (196.09,222.4) .. (184.99,222.4) .. controls (173.89,222.4) and (164.89,213.4) .. (164.89,202.3) -- cycle ;
\draw    (260,170) .. controls (263.43,178.04) and (245.34,215.94) .. (207.44,203.15) ;
\draw [shift={(205.09,202.3)}, rotate = 381.05] [fill={rgb, 255:red, 0; green, 0; blue, 0 }  ][line width=0.08]  [draw opacity=0] (8.93,-4.29) -- (0,0) -- (8.93,4.29) -- cycle    ;
\draw    (164.89,202.3) .. controls (129.18,217.48) and (115.72,201.75) .. (110.66,182.71) ;
\draw [shift={(110,180)}, rotate = 437.44] [fill={rgb, 255:red, 0; green, 0; blue, 0 }  ][line width=0.08]  [draw opacity=0] (8.93,-4.29) -- (0,0) -- (8.93,4.29) -- cycle    ;
\draw   (80,148.8) -- (125.5,148.8) -- (125.5,180) -- (80,180) -- cycle ;
\draw    (80,148.8) .. controls (68.9,142.43) and (58.27,139.41) .. (77.76,121.96) ;
\draw [shift={(80,120)}, rotate = 499.52] [fill={rgb, 255:red, 0; green, 0; blue, 0 }  ][line width=0.08]  [draw opacity=0] (8.93,-4.29) -- (0,0) -- (8.93,4.29) -- cycle    ;
\draw   (69.9,51.1) .. controls (69.9,40) and (78.9,31) .. (90,31) .. controls (101.1,31) and (110.1,40) .. (110.1,51.1) .. controls (110.1,62.2) and (101.1,71.2) .. (90,71.2) .. controls (78.9,71.2) and (69.9,62.2) .. (69.9,51.1) -- cycle ;
\draw    (196.5,47.2) -- (237.12,59.15) ;
\draw [shift={(240,60)}, rotate = 196.4] [fill={rgb, 255:red, 0; green, 0; blue, 0 }  ][line width=0.08]  [draw opacity=0] (8.93,-4.29) -- (0,0) -- (8.93,4.29) -- cycle    ;
\draw    (254.9,95.2) -- (259.66,137.02) ;
\draw [shift={(260,140)}, rotate = 263.51] [fill={rgb, 255:red, 0; green, 0; blue, 0 }  ][line width=0.08]  [draw opacity=0] (8.93,-4.29) -- (0,0) -- (8.93,4.29) -- cycle    ;
\draw    (80,106) .. controls (68.79,97.42) and (48.54,92.45) .. (77.66,77.19) ;
\draw [shift={(80,76)}, rotate = 513.51] [fill={rgb, 255:red, 0; green, 0; blue, 0 }  ][line width=0.08]  [draw opacity=0] (8.93,-4.29) -- (0,0) -- (8.93,4.29) -- cycle    ;
\draw    (110.1,51.1) -- (147,50.08) ;
\draw [shift={(150,50)}, rotate = 538.4200000000001] [fill={rgb, 255:red, 0; green, 0; blue, 0 }  ][line width=0.08]  [draw opacity=0] (8.93,-4.29) -- (0,0) -- (8.93,4.29) -- cycle    ;

\draw (166,35.4) node [anchor=north west][inner sep=0.75pt]    {$b_{1}$};
\draw (247,66.4) node [anchor=north west][inner sep=0.75pt]    {$a_{1}$};
\draw (241,147.2) node [anchor=north west][inner sep=0.75pt]    {$b_{2}$};
\draw (177.09,193.6) node [anchor=north west][inner sep=0.75pt]    {$a_{2}$};
\draw (95,156.2) node [anchor=north west][inner sep=0.75pt]    {$b_{3}$};
\draw (82,101) node [anchor=north west][inner sep=0.75pt]   [align=left] {$\displaystyle \vdots $};
\draw (82,81) node [anchor=north west][inner sep=0.75pt]   [align=left] {$\displaystyle \vdots $};
\draw (81,42.4) node [anchor=north west][inner sep=0.75pt]    {$a_{n}$};

\end{tikzpicture}
}
\caption{The trading graph of $P_{2}(n)$ is a single cycle consisting of all the agents and items from $A_{n} \cup I_{n}$. This is because in $P_{2}(n)$, for each $i \in [n-1]$, $a_{i}$ prefers only $b_{i+1}$ over $b_{i}$, and $a_{n}$ prefers only $b_{1}$ over $b_{n}$.}\label{P2G}
\end{figure}
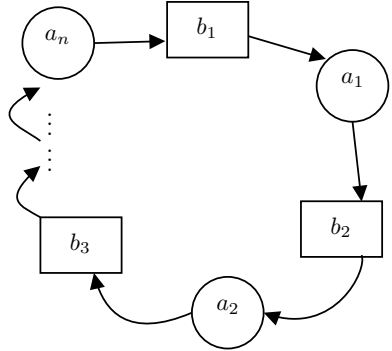

Informally speaking, we will construct the second preference profile, $P_{2}(n)$, so that its trading graph with respect to $p_{n}$ will contain a single trading cycle consisting of all the agents and items from $A_{n} \cup I_{n}$ (see \Ficref{P2G}). It is formally constructed as follows.
\begin{framed}[0.45\columnwidth]
\begin{itemize}
  \item $a_{1}\ $: $\ b_{2}\ >\ b_{1}$
  \item $a_{2}\ $: $\ b_{3}\ >\ b_{2}$
  \item $\ldots$
  \item $a_{n-1}\ $: $\ b_{n}\ >\ b_{n-1}$
  \item $a_{n}\ $: $\ b_{1}\ >\ b_{n}$
\end{itemize}
\end{framed}

\begin{observation}
\label{obs:coNPHObs1}
The only trading cycle in $P_{2}(n)$ with respect to $p_{n}$ is $(a_{1},b_{1},\ldots,a_{n},b_{n})$.
\end{observation}
\begin{proof}
For each $i \in [n-1]$, the only item that $a_{i}$ prefers over its assigned item $p_{n}(a_{i}) = b_{i}$ is $b_{i+1}$; and $a_{n}$ only prefers $b_{1}$ over its assigned item $p_{n}(a_{n}) = b_{n}$.\qed
\end{proof}
 
Intuitively, we defined the second preference profile to ensure that if there exists a subset of agents that admits trading cycles in both $P_{1}(G)$ and $P_{2}(n)$ with respect to $p_{n}$, this subset must be equal to $A_{n}$. Thus, by \Cocref{lemma:coNPHLemma2}, we will conclude that $G$ contains a Hamiltonian cycle.

\begin{lemma}\label{lemma:coNPHLemma3}
A set $K \subseteq A_{n}$ admits trading cycles in both $P_{1}(G)$ and $P_{2}(n)$ with respect to $p_{n}$ if and only if $K = A_{n}$ and $G$ contains a Hamiltonian cycle. 
\end{lemma}
\begin{proof}

$\Rightarrow$: Suppose that $K$ admits trading cycles in both $P_{1}(G)$ and $P_{2}(n)$ with respect to $p_{n}$. By \Obcref{obs:coNPHObs1}, we have that $K = A_{n}$. Thus, by \Cocref{lemma:coNPHLemma2}, $G$ contains a Hamiltonian cycle.

\smallskip
\noindent
$\Leftarrow$: Suppose that $K = A_{n}$ and that $G$ contains a Hamiltonian cycle. Then, by \Cocref{lemma:coNPHLemma2}, $K$ admits a trading cycle in $P_{1}(G)$, and by \Obcref{obs:coNPHObs1}, $K$ admits a trading cycle in $P_{2}(n)$.\qed
\end{proof}

We now use these results to prove that \vAAssign, \vDAssign, and \vCAssign\ are \paraCoH\ for the parameter $(n - k) + \ell + d$. We will rely on the result of Plesńik \cite{PLESNIK1979199}, who proved that \hc\ is \NPH\ on directed graphs with maximum degree $3$. In particular, the degree bound will help us bound the maximum length of a preference list in the preference profile $P_{1}(G)$.

\begin{theorem}\label{theorem:coNPH}
The problems \vAAssign, \vDAssign\ and \vCAssign\ are \CONPH\ when $k=n$, $\ell = 1$ ($\ell = 2$ for \vDAssign), $\alpha=1$, and $d = 3$.
\end{theorem}
\begin{proof}
We provide a polynomial reduction from \hc\ on directed graphs with maximum degree 3 (proved to be \NPH\ by Plesńik \cite{PLESNIK1979199}) to the complement of  \vAAssign\ (\coVAAssign) where $k=n$, $\alpha = \ell = 1$, and $d = 3$. We then explain why the reduction yields the same result for $\vCAssign$. Afterwards, we will extend the proof to have the same result also for $\vDAssign$.

Let $G = (V,E)$ be a directed graph of maximum degree 3 with $n$ vertices. We construct an instance of \coVAAssign\ consisting of one layer, $n = |V|$ agents and $n$ items as well, which is a \yes-instance\ if and only if $G$ contains a Hamiltonian cycle. By \Obcref{observation:notAOptimalA}, in order to prove that the resulting instance is a \yes-instance, we will show that the set of all agents admits a trading cycle in the only layer. Assume that $V = \lbrace v_{1},\ldots,v_{n} \rbrace$ is the vertex set of $G$. We construct the agent set $A_{n}$, the item set $I_{n}$ and the assignment $p_{n} : A_{n} \rightarrow I_{n}$ (as defined at the beginning of the section). The single layer contains the preference profile $P_{1}(G)$. The resulting instance of \coVAAssign\ is $D = (A_{n},I_{n},P_{1}(G),\alpha = 1, k = n , p_{n})$ which can be clearly constructed in polynomial time. Since the maximum degree of $G$ is 3, by the definition of $P_{1}(G)$, we have that the maximum length of a preference list in the resulting instance is $d = 3$. We now prove that $G$ contains a Hamiltonian cycle if and only if $D$ is a \yes-instance\ of  \coVAAssign.

\smallskip
\noindent
$\Rightarrow$: Assume that $G$ contains a Hamiltonian cycle. By \Cocref{lemma:coNPHLemma2}, the set $A_{n}$ admits a trading cycle in $P_{1}(G)$. Then, by \Obcref{observation:notAOptimalA}, $D$ is a \yes-instance.

\smallskip
\noindent
$\Leftarrow$: Assume that $D$ is a \yes-instance. Then, by \Obcref{observation:notAOptimalA}, $A_{n}$ admits a trading cycle in the only layer $P_{1}(G)$. By \Cocref{lemma:coNPHLemma2}, $G$ contains a Hamiltonian cycle.

By \Cocref{col:nEquiv}, \vCAssign\ is equivalent to \vAAssign\ when $k=n$. Therefore, the result holds for \vCAssign\ as well.

\smallskip
\noindent\textbf{\vDAssign.} We extend the proof to the problem \vDAssign. We add another layer to the constructed instance containing the preference profile $P_{2}(n)$. By \Obcref{obs:coNPHObs1}, since the only trading cycle in $P_{2}(n)$ contains exactly $n$ agents (by \Obcref{obs:coNPHObs1}), we conclude that $p_{n}$ is $(k,\alpha)$-\aoptimal\ for each $k < n$ (since each subset of less than $n$ agents does not admit trading cycles in $P_{2}(n)$). So, by \Obcref{lemma:aAndDEq}, $p_{n}$ is $(n,\alpha)$-\aoptimal\ if and only if it is $(n,\alpha)$-\doptimal. Thus, the reduction can be equivalently a polynomial reduction to \vDAssign. \qed
\end{proof}

Since $\ell+d+(n-k)$ is bounded by a constant in the resulting instances, we conclude the following.
\begin{corollary}\label{col:ellPlusKC}
The problems \vAAssign, \vDAssign\ and \vCAssign\ are \paraCoH\ with respect to the parameter $\ell+ d + (n-k)$.
\end{corollary}
\section{Non-Existence of Polynomial Kernels}
\label{sec:part2Kernelization}
In this section, We prove that the three problems are unlikely to admit polynomial kernels with respect to $n + m + \alpha$ and $n + m + (\ell - \alpha)$. So, considering $\alpha$ or $(\ell - \alpha)$ rather than $\ell$, even while considering the larger parameter $n + m$ rather than $\nalloc$, yields negative results. So, our classification is complete in this sense. To prove this, we will rely on the results in Section \ref{sec:conphardness} to provide an OR-cross-composition and an AND-cross-composition from $\hc$ on directed graphs with maximum degree 3 to the complements of the problems. We first define the polynomial equivalence relation $\RR$ that we will use. We say that two directed graphs are equivalent with respect to $\RR$ if they share the same number of vertices. This is known to be a polynomial equivalence relation, yet we present the proof for completeness.

\begin{proposition}[Folklore]\label{lemma:polyEqRelation}
$\RR$ is a polynomial equivalence relation.
\end{proposition} 
\begin{proof}
It is clear that $\RR$ is an equivalence relation. Moreover, if we restrict $\RR$ on encodings of graphs of length at most $n$, we have that $\RR$ admits at most $n$ equivalence classes since every such graph can contain at most $n$ vertices. Given two graphs  $G_{1}$ and $G_{2}$, it can clearly be verified in time $\OO(|G_{1}|+|G_{2}|)$ whether the graphs are equivalent with respect to $\RR$ by checking if they have the same number of vertices.\qed
\end{proof}

\begin{theorem}\label{theorem:cc1}
The problems \vAAssign, \vDAssign, and \vCAssign\ do not admit polynomial kernels with respect to $n + m + \alpha$ and $n + m + (\ell - \alpha)$, unless \NP$\subseteq $\coNPpoly. 
\end{theorem}
\begin{proof}
We provide a cross-composition from \hc\ on directed graphs with maximum degree 3 to to complement of \vAAssign\  (\coVAAssign). So, this will prove that \vAAssign\ does not admit a polynomial kernel with respect to these parameters. It can be easily shown that a problem admits a polynomial kernel if and only if its complement admits a polynomial kernel (the same kernel can be used for both problems). So, for the parameter $n + m + \alpha$, we will treat the cross-composition as an AND-cross-composition, and for the parameter $n + m + (\ell - \alpha)$, we will treat it as an OR-cross-composition. Similarly to the proof of \Thcref{theorem:coNPH}, we will explain why they are also cross-compositions to \coVCAssign. After that, we will explain how to extend the cross-compositions to derive the same results on \vDAssign.

\begin{figure}[t]
\center \scalebox{0.89}{\tikzset{every picture/.style={line width=0.75pt}} 

\begin{tikzpicture}[x=0.75pt,y=0.75pt,yscale=-1,xscale=1]

\draw   (251.5,97.52) -- (259.75,97.52) -- (259.75,65) -- (271.75,65) -- (271.75,97.52) -- (280,97.52) -- (265.75,113) -- cycle ;
\draw   (103,138) -- (168.5,138) -- (168.5,214.2) -- (103,214.2) -- cycle ;
\draw   (187.5,138) -- (253,138) -- (253,214.2) -- (187.5,214.2) -- cycle ;
\draw   (305,138) -- (370.5,138) -- (370.5,214.2) -- (305,214.2) -- cycle ;
\draw   (395,138) -- (460.5,138) -- (460.5,214.2) -- (395,214.2) -- cycle ;

\draw    (267.5, 38) circle [x radius= 95.46, y radius= 15.56]   ;
\draw (201,29) node [anchor=north west][inner sep=0.75pt]   [align=left] {$\displaystyle G_{1} ,G_{2} ,G_{3} ,\dotsc ,G_{t}$};
\draw (109,168.4) node [anchor=north west][inner sep=0.75pt]    {$P_{1}( G_{1}) \ $};
\draw (193.5,168.4) node [anchor=north west][inner sep=0.75pt]    {$P_{1}( G_{2}) \ $};
\draw (256,169) node [anchor=north west][inner sep=0.75pt]   [align=left] {$\displaystyle \dotsc \dotsc $};
\draw (310,168.4) node [anchor=north west][inner sep=0.75pt]    {$P_{1}( G_{t}) \ $};
\draw (110,116) node [anchor=north west][inner sep=0.75pt]   [align=left] {Layer $\displaystyle 1$};
\draw (194.5,116.2) node [anchor=north west][inner sep=0.75pt]   [align=left] {Layer $\displaystyle 2$};
\draw (313,116) node [anchor=north west][inner sep=0.75pt]   [align=left] {Layer $\displaystyle t$};
\draw (400,168.4) node [anchor=north west][inner sep=0.75pt]    {$P_{2}( n) \ $};
\draw (397,116) node [anchor=north west][inner sep=0.75pt]   [align=left] {Layer $\displaystyle t+1$ \\};
\draw (390,221) node [anchor=north west][inner sep=0.75pt]   [align=left] {(for Verify-UOA)};

\end{tikzpicture}}
\caption{The cross-compositions construct an instance containing $P_{1}(G_{i})$ in layer $i$, for each $i \in [t]$. The AND-cross-composition for \vDAssign\ constructs an additional layer containing $P_{2}(n)$.}\label{CC1Figure1}
\end{figure}
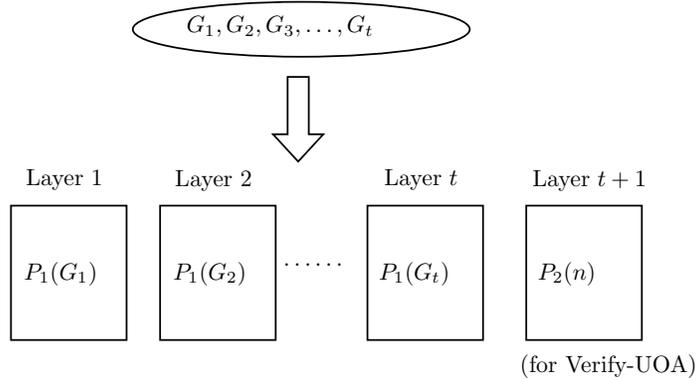

Consider $t$ directed graphs with maximum degree 3 $(G_{1} = (V_{1},E_{1}),\ldots,G_{t}=(V_{t},E_{t}))$ that are equivalent with respect to $\RR$. We assume that each graph $G_{i}$ contains $n$ vertices, and $V_{i} = V = \lbrace v_{1},\ldots,v_{n} \rbrace$ for each $i \in [t]$. Informally speaking, the cross-compositions are ``compositions'' of the instances created by applying the reduction in the proof of \Thcref{theorem:coNPH} on each input graph (see \Ficref{CC1Figure1}). The agent set of the constructed instance is $A_{n} = \lbrace a_{1},\ldots,a_{n} \rbrace$, the item set is $I_{n} = \lbrace b_{1},\ldots,b_{n} \rbrace$, and there are $t$ layers, such that for each $i \in [t]$, layer $i$ contains the preference profile $P_{1}(G_{i})$. We set $k = n$, and we also use the assignment $p = p_{n}$. Notice that the resulting instance can clearly be constructed in time $poly(\sum_{i=1}^{t}{|G_{i}|})$. We complete the constructions for each parameter separately.

\smallskip
\noindent\textbf{The parameter $n + m + \alpha$.} We treat this reduction as an AND-cross-composition, and we set $\alpha = 1$. We prove that the set of agents, $A_{n}$, admits trading cycles in all the layers of the constructed instance with respect to $p_{n}$ if and only if all the graphs $G_{i}$, $i \in [t]$, admit a Hamiltonian cycle.

\smallskip
\noindent $\Rightarrow$: Suppose that the agents in $A_{n}$ admit trading cycles in all the layers of the constructed instance with respect to $p_{n}$. This implies that the agents in $A_{n}$ admit trading cycles in $P_{1}(G_{i})$ for all $i \in [t]$. By \Cocref{lemma:coNPHLemma2}, each $G_{i}$ admits a Hamiltonian cycle.

\smallskip 
\noindent $\Leftarrow$: Suppose that all the graphs $G_{i}$ admit a Hamiltonian cycle. Then, by \Cocref{lemma:coNPHLemma2}, the agents in $A_{n}$ admit trading cycles in $P_{1}(G_{i})$ for all $i \in [t]$. Then, the agents in $A_{n}$ admit trading cycles in all the layers of the constructed instance. 

Due to \Obcref{observation:notAOptimalA}, this proves that the resulting instance is a \yes-instance\ of \coVAAssign\ if and only if all the input instances are \yes-instance.  We also have that $n + m + \alpha \le 2n + 1 = \OO(\max_{i=1}^{t}{|G_{i}|})$. Thus, by \Prcref{prop:noKern}, \vAAssign\ does not admit a polynomial kernel with respect to the parameter $n + m + \alpha$, unless \NP$\subseteq $\coNPpoly. Since we set $k = n$, by \Cocref{col:nEquiv}, the resulting instance can also represent an equivalent instance of \coVCAssign. Then, the same result also holds for \vCAssign.

\smallskip
\noindent\textbf{The parameter $n + m + (\ell - \alpha)$.} For this parameter, we treat the reduction as an OR-cross-composition, and we set $\alpha = \ell = t$. We will prove that the resulting instance is a \yesinstance\ of \coVAAssign\ if and only if there exists $i \in [t]$, such that $G_{i}$ contains a Hamiltonian cycle. By \Obcref{observation:notAOptimalA}, we need to show that the set $A_{n}$ admits trading cycles in $\ell - \alpha + 1 = 1$ layer if and only if there exists $i \in [t]$ such that $G_{i}$ is a \yesinstance.

$\Rightarrow$: Suppose that $A_{n}$ admits a trading cycle in layer $i \in [t]$. Thus, $A_{n}$ admits a trading cycle in $P_{1}(G_{i})$. Then, by \Cocref{lemma:coNPHLemma2}, $G_{i}$ contains a Hamiltonian cycle.

$\Leftarrow$: Suppose that there exists $i \in [t]$ such that $G_{i}$ contains a Hamiltonian cycle. Then, by \Cocref{lemma:coNPHLemma2}, $A_{n}$ admits a trading cycle in $P_{1}(G_{i})$, which appears in layer $i$.

This proves the correctness of the OR-cross-composition. Notice that $n + m + (\ell - \alpha) = 2n = poly(\max_{i=1}^{t}{|G_{i}|})$ for the constructed instance. Then, by \Prcref{prop:noKern}, \vAAssign\ does not admit a polynomial kernel with respect to $n + m + (\ell - \alpha)$, unless \NP$\subseteq $\coNPpoly. Since we set $k = n$, by \Cocref{col:nEquiv}, the resulting instance can also represent an equivalent instance of \coVCAssign. Thus, the result holds for \vCAssign\ as well.

\smallskip
\noindent\textbf{\vDAssign.} In order to adapt the first AND-cross-composition to prove the same result for \vDAssign, we need to control the size of the subsets admitting trading cycles in all the layers. To do so, similarly to the proof of \Thcref{theorem:coNPH}, we add another layer to the constructed instance, which contains the preference profile $P_{2}(n)$. By \Lecref{lemma:coNPHLemma3}, a subset $K \subseteq A_{n}$ admits trading cycles in all the layers if and only if $K = A_{n}$ and each $G_{i}$ contains a Hamiltonian cycle. Namely, $p_{n}$ is $(n,1)$-\aoptimal\ if and only if it is $(n,1)$-\doptimal. Thus, we have that the resulting instance represents an equivalent instance of \vDAssign, and the result holds for this problem as well.

We now adapt the second OR-cross-composition to prove the same result for \vDAssign. On a high level, we insert $2 (\floor{\log{t}}+1)$ additional agents, $2 (\floor{\log{t}}+1)$ additional items, and $t$ additional layers. We insert into each layer of the original instance a unique set of preference lists that correspond to the new $2 (\floor{\log{t}}+1)$ agents, and we add another layer after it containing $P_{2}(n)$, appended with the same set of preference lists for the new agents.

Formally, we define the new agent set $C = \lbrace c_{i} | i \in [\floor{\log{t}}+1] \rbrace \cup \lbrace \overline{c_{i}} | i \in [\floor{\log{t}}+1] \rbrace$, and the new item set $D = \lbrace d_{i} | i \in [\floor{\log{t}}+1] \rbrace \cup \lbrace \overline{d_{i}} | i \in [\floor{\log{t}}+1] \rbrace$. The agent set and the item set of the constructed instance are $A = A_{n} \cup C$ and $I = I_{n} \cup D$, respectively. We define the assignment $p : A \rightarrow I$ by $p(a_{i}) = b_{i}$ for each $i \in [n]$; and $p(c_{i}) = d_{i}$, $p(\overline{c_{i}}) = \overline{d_{i}}$ for each $i \in [\floor{\log{t}}+1]$. Notice that the restriction of $p$ to $A_{n}$ is equal to $p_{n}$. We now construct $2t$ new layers as follows. Informally speaking, each input graph $G_{i}$ will have two corresponding layers, $2i-1$ and $2i$, which are compositions of the preference profiles $P_{1}(G_{i})$ or $P_{2}(n)$ for each graph $G_{i}$ (defined in Section \ref{sec:conphardness}) together with $2 (\floor{\log{t}}+1)$ unique preference lists for the agents in $C$. Intuitively, the goal of the agents and items in $C \cup D$ is to ensure that for each $i \in [t]$, there is a unique subset $C' \subseteq C$ of size $\floor{\log{t}}+1$ that is part of trading cycles in both layers $2i - 1$ and $2i$. This will imply that if there exists a subset of agents from $A_{n} \cup C$ admitting trading cycles in both of these layers, then this subset is a unique subset of size exactly $n + \floor{\log{t}}+1$ which does not admit trading cycles in the rest of the layers. If $G_{i}$ contains a Hamiltonian cycle, then we will have in layers $2i-1$ and $2i$ similar trading cycles as in $P_{1}(G_{i})$ and $P_{2}(n)$, but appended with a chain of $\floor{\log{t}}+1$ agents from $C$. For $i \in [t]$, we denote by $i[j]$ the $j$'th bit in the binary representation of $i$, for each $j \in [\floor{\log{t}}+1]$. Denote $\cc_{j}^{i} = c_{j}$ if $i[j] = 1$ and $\cc_{j}^{i} = \overline{c_{j}}$ if $i[j] = 0$. Similarly, $\dd_{j}^{i} = d_{j}$ if $i[j] = 1$ and $\dd_{j}^{i} = \overline{d_{j}}$ if $i[j] = 0$. Notice that $p(\cc_{j}^{i}) = \dd_{j}^{i}$. Assume that $\overline{\overline{c_{j}}} = c_{j}$ for each $j \in [\floor{\log{t}}+1]$.

We now construct the layers formally. For each $i \in [t]$, we create two preference profiles, $Q_{2i-1}$ and $Q_{2i}$, that appear in layers $2i-1$ and $2i$, respectively. The preference profile $Q_{2i-1}$ extends (and slightly modifies) $P_{1}(G_{i})$ as follows.

\begin{framed}[1\columnwidth]
\begin{itemize}
  \item $a_{i}\ $: $\lbrace b_{j} | (v_{i},v_{j}) \in E_{i} \rbrace\ \text{(in arbitrary order)}\ >\ b_{i}\ \ \forall i\in[n - 1]$
  \item $a_{n}\ $: $\ \dd_{1}^{i}\ >\ b_{n}$
  \item $\cc_{j}^{i}\ $: $\ \dd_{j+1}^{i}\ >\ \dd_{j}^{i}\ \ \forall j \in [\floor{\log{t}}]$
  \item $\cc_{\floor{\log{t}}+1}^{i}\ $: $\ \lbrace b_{j} | (v_{n},v_{j}) \in E_{i} \rbrace\ \text{(in arbitrary order)}\ >\ \dd_{\floor{\log{t}}+1}^{i}$
  \item $\overline{\cc_{j}^{i}} :\ \overline{\dd_{j}^{i}}\ \ \forall j \in [\floor{\log{t}}+1]$
\end{itemize}
\end{framed}

The preference profile $Q_{2i}$ extends (and slightly modifies) $P_{2}(n)$ as follows.

\begin{framed}[.7\columnwidth]
\begin{itemize}
  \item $a_{1}\ $: $\ b_{2}\ >\ b_{1}$
  \item $a_{2}\ $: $\ b_{3}\ >\ b_{2}$
  \item $\ldots$
  \item $a_{n-1}\ $: $\ b_{n}\ >\ b_{n-1}$
  \item $a_{n}\ $: $\ \dd_{1}^{i} >\ b_{n}$
    \item $\cc_{j}^{i}\ $: $\dd_{j+1}^{i}\ >\ \dd_{j}^{i}\ \ \forall j \in [\floor{\log{t}}]$
  \item $\cc_{\floor{\log{t}}+1}^{i}\ $: $\ b_{1}\ >\ \dd_{\floor{\log{t}}+1}^{i}$
  \item $\overline{\cc_{j}^{i}} : \overline{\dd_{j}^{i}}\ \ \forall j \in [\floor{\log{t}}+1]$
\end{itemize}
\end{framed}

\begin{figure*}[t]
\center \scalebox{.82}{\input{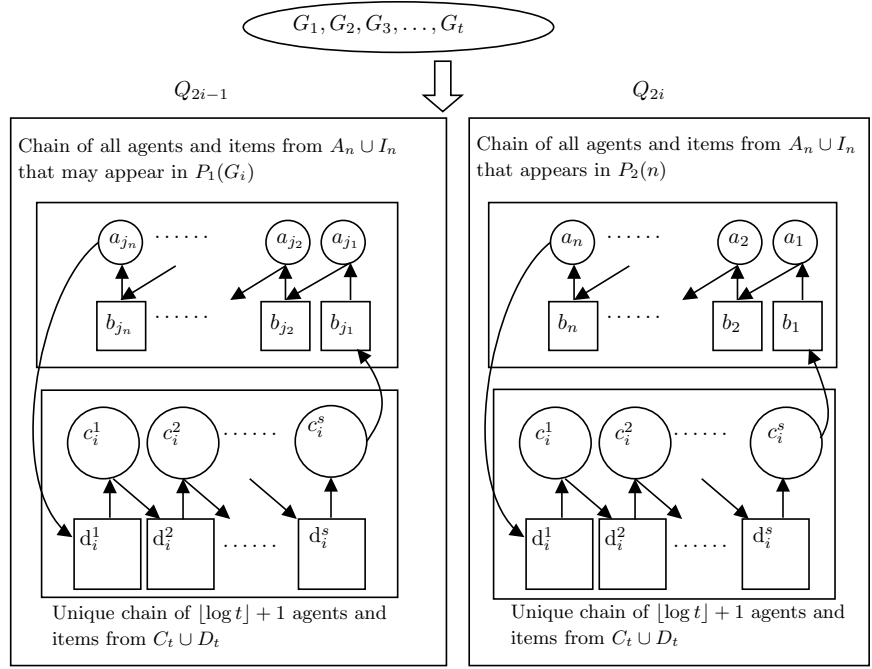}}
\caption{$Q_{2i-1}$ is a composition of a modification of $P_{1}(G_{i})$ with $2s$ unique preference lists for the agents in $C$ (where $s = \floor{\log{t}} + 1$). $(a_{j_{1}},b_{j_{1}},\ldots,a_{j_{n}},b_{j_{n}})$ is a trading cycle in $P_{1}(G_{i})$ where $j_{n} = n$ (this occurs only when $G_{i}$ is a \yes-instance) if and only if $(a_{j_{1}},b_{j_{1}},\ldots,a_{j_{n}},b_{j_{n}},\cc_{1}^{i},\dd_{1}^{i},\ldots,\cc_{\floor{\log{t}}+1}^{i},\dd_{\floor{\log{t}}+1}^{i})$ is a trading cycle in $Q_{2i-1}$. The trading graph of $Q_{2i}$ contains the trading cycle $(a_{1},b_{1},\ldots,a_{n},b_{n},\cc_{1}^{i},\dd_{1}^{i},\ldots,\cc_{\floor{\log{t}}+1}^{i},\dd_{\floor{\log{t}}+1}^{i})$, which corresponds to the trading cycle $(a_{1},b_{1},\ldots,a_{n},b_{n})$ in $P_{2}(n)$.}
\label{CC2Figure1}
\end{figure*}

The possible trading cycles in $Q_{2i-1}$ and $Q_{2i}$ are shown in \Ficref{CC2Figure1}. We finally set $k = n + \floor{\log{t}}+1$ and $\alpha = \ell - 1 = 2t - 1$ (then $\ell - \alpha + 1 = 2$).
Let us now claim the following, to relate trading cycles in the original and new preference profiles $P_{1}(G_{i})$ and $Q_{2i-1}$

\begin{myclaim}\label{claim:or-cc-claim1}
Let $i\in[t]$ and let $W = (a_{j_{1}},b_{j_{1}},\ldots,a_{j_{n}},b_{j_{n}})$ be a sequence of agents and their assigned items with respect to $p_{n}$ such that $j_{n} = n$. Then, $W$ is a trading cycle in $P_{1}(G_{i})$ with respect to $p_{n}$ if and only if $W' = (a_{j_{1}},b_{j_{1}},\ldots,a_{j_{n}},b_{j_{n}},\cc_{1}^{i},\dd_{1}^{i},\ldots,\cc_{\floor{\log{t}}+1}^{i},\dd_{\floor{\log{t}}+1}^{i})$ is a trading cycle in $Q_{2i-1}$ with respect to $p$ (see \Ficref{CC2Figure1}).
\end{myclaim}

\begin{proof}

\smallskip
\noindent
$\Rightarrow$: Assume that $W$ is a trading cycle in $P_{1}(G_{i})$ with respect to $p_{n}$. By the construction of $Q_{2i-1}$, observe that the trading graph of $Q_{2i-1}$ contains the paths $b_{i_{1}} \rightarrow a_{i_{1}} \rightarrow \ldots \rightarrow b_{i_{n}} \rightarrow a_{i_{n}}$, and $b_{n} \rightarrow a_{n} \rightarrow \dd_{1}^{i} \rightarrow \cc_{1}^{i} \rightarrow \ldots \rightarrow \dd_{\floor{\log{t}}+1 }^{i} \rightarrow \cc_{\floor{\log{t}}+1 }^{i} \rightarrow b_{i_{1}}$. By concatenating these two paths, we have that $Q_{2i-1}$ contains the trading cycle $W'$.

\smallskip
\noindent
$\Leftarrow$: Assume that $W'$ is a trading cycle $Q_{2i-1}$ with respect to $p$. By the construction of $Q_{2i-1}$, note that the trading graph of $P_{1}(G_{i})$ contains the path $b_{i_{1}} \rightarrow a_{i_{1}} \rightarrow \ldots \rightarrow b_{i_{n}} \rightarrow a_{i_{n}}$. Since $\cc_{\floor{\log{t}}+1}^{i}$ prefers all the items that correspond to neighbors of $v_{n}$ to which $v_{n}$ points, we have that $(v_{n},v_{i_{1}}) \in E_{i}$, and $a_{n}$ prefers $b_{i_{1}}$ over $b_{n}$ in $P_{1}(G_{i})$. Thus, $P_{1}(G_{i})$ contains the trading cycle $W$ with respect to $p_{n}$.\qed
\end{proof}

We proceed to consider trading cycles in the new profiles, now for $Q_{2i}$.

\begin{myclaim}\label{claim:or-cc-claim2}
For each $i \in [t]$, the only trading cycle in $Q_{2i}$ with respect to $p$ is $(a_{1},b_{1},\ldots,a_{n},b_{n},\cc_{1}^{i},\dd_{1}^{i},\ldots,\cc_{\floor{\log{t}}+1}^{i} , \dd_{\floor{\log{t}}+1}^{i})$.
\end{myclaim}
\begin{proof}
By \Obcref{obs:coNPHObs1}, the only trading cycle in $P_{2}(n)$ with respect to $p_{n}$ is $(a_{1},b_{1},\ldots,a_{n},b_{n})$. Observe that the trading graph of $Q_{2i}$ contains the paths $b_{1} \rightarrow a_{1} \rightarrow \ldots \rightarrow b_{n} \rightarrow a_{n}$ and $a_{n} \rightarrow \dd_{1}^{i} \rightarrow \cc_{1}^{i} \rightarrow \ldots \rightarrow \dd_{\floor{\log{t}}+1}^{i} \rightarrow \cc_{\floor{\log{t}}+1}^{i}$, and the only trading cycle is constructed by concatenating these two paths.\qed 
\end{proof}

By these results, we conclude the following.

\begin{myclaim}\label{claim:or-cc1}
Let $i \in [t]$. The only subset of agents that can admit trading cycles in both $Q_{2i-1}$ and $Q_{2i}$ with respect to $p$ is $A_{n} \cup \lbrace \cc_{j}^{i} | j \in [\floor{\log{t}}+1] \rbrace$.
\end{myclaim}

\begin{myclaim}\label{claim:or-cc2}
Let $i,j \in [t]$ such that $i \neq j$, and let $K_{i},K_{j} \subseteq A$. Suppose that $K_{i}$ admits a trading cycle in $Q_{2i-1}$ or in $Q_{2i}$ and suppose that $K_{j}$ admits a trading cycle in $Q_{2j-1}$ or in $Q_{2j}$. Then, $K_{i} \neq K_{j}$.
\end{myclaim}

We now prove that the resulting instance is a \yes-instance\ of \coVDAssign\ if and only if there exists $i \in [t]$ such that $G_{i}$ admits a Hamiltonian cycle. By \Obcref{observation:notAOptimalA}, we will show that there exists subset $K \subseteq A$ with $|K| \leq n + \floor{\log{t}}+1$ such that there exist two layers $i,j$ where the agents in $K$ admit trading cycles if and only if there exists $i \in [t]$ such that $G_{i}$ admits a Hamiltonian cycle. 

\smallskip
\noindent
$\Rightarrow$: Suppose that there exists $K \subseteq A$, with $|K| \leq n + \floor{\log{t}}+1$, and $i , j \in [t]$ such that $i \neq j$, and the agents in $K$ admit trading cycles in both layers $i,j$. By Claims \ref{claim:or-cc1} and \ref{claim:or-cc2}, we have that there exists $\tilde{i} \in [t]$ such that $i = 2\tilde{i} - 1$, $j = 2\tilde{i}$, and $K = A_{n} \cup \lbrace \cc_{j}^{\tilde{i}} | j \in [\floor{\log{t}}+1] \rbrace$. By \Clcref{claim:or-cc-claim1}, $P_{1}(G_{\tilde{i}})$ contains a trading cycle that contains all the agents in $A_{n}$. Moreover, by \Obcref{obs:coNPHObs1}, notice that $P_{2}(n)$ also contains a trading cycle over all the agents from $A_{n}$. Then, by \Cocref{lemma:coNPHLemma2}, $G_{i}$ admits a Hamiltonian cycle. 

\vspace{0.5em}
\noindent $\Leftarrow$: Suppose that there exists $i \in [t]$ such that $G_{i}$ admits a Hamiltonian cycle. Then, by \Lecref{lemma:coNPHLemma1}, $P_{1}(G_{i})$ contains some trading cycle $(a_{i_{1}},b_{i_{1}},\ldots,a_{i_{n}},b_{i_{n}})$ over all the agents in $A_{n}$. By \Clcref{claim:or-cc-claim1}, $Q_{2i-1}$ contains a trading cycle over all the agents from $K = A_{n} \cup \lbrace \cc_{j}^{i} | j \in [\floor{\log{t}}+1] \rbrace$. In addition, by \Clcref{claim:or-cc-claim2}, $K$ also admits a trading cycle in $Q_{2i}$ with respect to $p$. Then, we have that $K$ admits trading cycles in both layers $2i-1$ and $2i$, thus the resulting instance is a \yes-instance\ by \Obcref{observation:notAOptimalA}. 
The construction can be clearly be implemented in polynomial time in $\sum_{i=1}^{t}{|G_{i}|}$. Notice that $n + m + (\ell - \alpha) = 2n + 4(\floor{\log{t}}+1) + 1  = 2n + 4\floor{\log{t}} + 5 = poly(\max_{i=1}^{t}{|G_{i}|} + \log(t))$ for the constructed instance. So, by \Prcref{prop:noKern}, \vDAssign\ does not admit a polynomial kernel with respect to $n + m + (\ell - \alpha)$, unless \NP$\subseteq $\coNPpoly.\qed
\end{proof}
\vspace{-1em}

\section{Conclusion and Future Research}
In this paper, we introduced a generalization of the verification variant of the general assignment problem where each agent is equipped with multiple incomplete preference lists. We defined three natural concepts of optimality, we considered several natural parameters and we presented an almost comprehensive picture of the parameterized complexity of the corresponding problems with respect to them. We proved that the problems are \paraCoH\ with respect to $\ell + d + (n - k)$. We also proved that the three problems admit polynomial kernels when parameterized by $\nalloc + \ell$, but that they are unlikely to admit polynomial kernels with respect to $n + m + \alpha$ and $n + m + (\ell - \alpha)$. Additionally, we proved that the problems are \coWOH\ with respect to $k + \ell$. However, we showed that \vAAssign\ and \vDAssign\ admit \XP\ algorithms with respect to $k$, and even \FPT\ algorithms with respect to $k + d$. We also provided $\OO^{*}(2^{\nalloc})$-time algorithms for the three problems. This proved that the problems are \FPT\ with respect to the parameters $\nalloc$, $n$, and $m$. Still, two questions remained open: 
\begin{enumerate}
\item Is it possible to obtain an $\OO^{*}((2 - \varepsilon)^{\nalloc})$-time algorithm, for some fixed $\varepsilon > 0$, for each one of the problems?
\item Does \vDAssign\ admit \XP\ algorithms with respect to the parameters $k$, $k + \ell$, and $k + d$?
\end{enumerate}

\smallskip
\noindent{\bf Additional Directions for Future Research.}
Continuing our research, it may be interesting to consider a new concept of optimality: We introduce the new notion of  {\em $(k,\alpha)$-\boptimality}, which weakens the notion of $(k,\alpha)$-optimality as follows. Consider an instance $(A,I,P_{1},\dots,P_{\ell},k,\alpha,p)$ of $\vAAssignFull$ for which $p$ is not $(k,\alpha)$-\aoptimal. Thus, there exists a group of agents $K$ of size $k$ which admits trading cycles (where agents may appear in any order) in some $\ell - \alpha + 1$ layers. Since the trading cycles in these layers are not necessarily the same, there may not exist one ``strategy'' that solves all these conflicts and improves the status of the agents in $K$ at once. In particular, if the agents in $K$ perform a possible beneficial trade in one layer, their status may get worse in other layers. Thus, one can claim that $p$ can be ``optimal'' since each ``small'' group cannot benefit in some $\ell - \alpha + 1$ layers in parallel. Informally, $(k,\alpha)$-\boptimality\ considers the {\em order} of the trading cycles and requires that for each subset of agents of size $k$, there exist some $\alpha$ layers where the agents in the subset cannot perform {\em the same} beneficial trade. It is formally defined as follows:

\begin{definition} [{\em $(k,\alpha)$-\boptimality}] \label{def:BOptimality}
An assignment $p$ is {\em $(k,\alpha)$-\boptimal} for an instance $(A,I,P_{1},\ldots,P_{\ell})$ if it satisfies the following conditions:
\begin{enumerate}
\item For each subset of agents $K \subseteq A$ such that $|K| = k$, there exist $\alpha$ layers $i_{1},\ldots,i_{\alpha}$ such that there does not exist a trading cycle $C$ containing all the agents in $K$ (with no additional agents) such that $C$ appears in layer $i_{j}$, for each $j \in [\alpha]$.
\item If $k = 1$, then for each $a \in A$, there exist $\alpha$ layers $i_{1},\ldots,i_{\alpha}$ such that there does not exist a self loop $L$ such that $a$ admits $L$ in layer $i_{j}$, for each $j \in [\alpha]$.
\end{enumerate}
\end{definition}

Thus, as a direction for future research we propose to study the new decision and verification problems: \bAssignFull\ and \vBAssignFull\ (which correspond to the notion of $(k,\alpha)$-\boptimality). 

Another direction is to extend the optimality notions to support more complicated ``inter-list interactions''. To explain this, suppose that some assignment $p$ is not $(k, \alpha)$-\boptimal. So, there exists a subset of $k$ agents and $\alpha$ layers where they can trade, using the same trading cycle, and benefit. However, performing the trading cycle may make the assignment worse for them in a large number (potentially $\ell - \alpha$) of other layers. So, perhaps such assignments should still be considered optimal. 

We remark that no notion of optimality is better or worse, but the choice depends on the scenario at hand. For example, will a subset of $k$ agents rebel if it finds many layers where it is dissatisfied, or only if it actually has a strategy that improves its situation? More philosophically, how do we know if our assignment is good or bad? From a public opinion point of view, the unordered variant may make more sense, but from a practitioner's point of view (who should actually improve an assignment if need be), the ordered variant might make more sense. We also remark that when $k = n$, the unordered version corresponds to global optimality, while the ordered version does not.

Another direction is to consider weighted versions of the problems. In this paper, we considered the basic ``unweighed'' model of the problems (since this is the first study of this kind). That is, all the criteria (layers) have the same importance. There are cases where some criteria may have higher importance than others, and we would like to give them a higher weight. A straightforward way to model these cases is by having several copies of layers. However, if weights are high and varied, this might lead to inefficiency.

Another approach is to generalize our model and concepts to allow ties in the preference lists. It is interesting to see how adding this feature will affect our results.

Lastly, we suggest to test our results practically, i.e.~implementing the algorithms for the problems, and testing them on real data sets. Besides the analysis of running times in practice, we find it interesting to see how much effect does using different notions of optimality has on the solutions, in particular, how do the solutions vary.
\newpage
\bibliographystyle{splncs04}
\noindent\bibliography{Refs}
\end{document}